\documentclass[11pt]{article}
\usepackage[margin=1in]{geometry}
\usepackage[
    backend=biber,
    style=numeric,
    natbib=true,
    giveninits=true,
    uniquename=init
]{biblatex}
\usepackage{booktabs, makecell}
\usepackage{pdflscape}
\usepackage{rotating}
\usepackage{amsmath, amsfonts, amssymb, amsthm, mathtools, subcaption}
\usepackage{tikz}
\usepackage[most]{tcolorbox}
\usepackage{thmtools}
\usepackage{thm-restate}
\usepackage{paralist}
\usetikzlibrary{shapes.callouts}

\newcommand{\sr}[1]{\text{sign-rank}\!\left(#1\right)}
\newcommand{\MP}{\operatorname{MP}}

\newcommand{\PM}{\operatorname{PM}}
\newcommand{\nand}{\operatorname{NAND}}
\usepackage{expl3}
  
 \newcommand{\cD}{\mathcal{D}}

\newcommand{\cQ}{\mathcal{Q}}

\newcommand{\eps}{\varepsilon}
\newcommand{\rank}{\textrm{rank~}}
\newcommand{\NAND}{\textrm{NAND~}}
\newcommand{\signrank}{\textrm{signrank~}}

\newtheorem{lemma}{Lemma}

\theoremstyle{definition}

\usepackage{pifont}

\newcommand{\cmark}{\ding{51}}
\newcommand{\xmark}{\ding{55}}
\tikzset{
    mybubble/.style={
        draw,
        fill=blue!20,
        rectangle callout,
        callout absolute pointer={(0,-1)},
        rounded corners,
        inner sep=10pt,
        font=\sffamily
    }
}

\tcbset{
  exbox/.style={
    width=\linewidth,
    boxsep=0.4mm,
    left=0.5mm,right=0.5mm,top=0.4mm,bottom=0.4mm,
    arc=0.8mm,
    colbacktitle=white,
    enhanced
  }
}
\title{Retrieval Needs Multivectors: An Exponential Separation}

\author{Mihir Agarwal, Viraj Agrawal, Sabyasachi Basu, \\ Ankit Garg, Kirankumar Shiragur \\
  \footnotesize \texttt{\{t-miagarwal, t-viragrawal, sabyasachi.basu, garga, kshiragur\} }@microsoft.com
  \small \\ Microsoft Research India
}

\date{}

\begin{document}

\maketitle

\begin{abstract}
Recent works have highlighted the expressive limitations of embedding based retrieval models through both theoretical analyses and challenging benchmarks such as LIMIT. While multi-vector embeddings consistently outperform single-vector embeddings, the precise representational gap between them remains poorly understood. In this work, following Jayaram's work, we provide the first explicit family of query and document sets, together with their relevance matrices, for which single-vector embeddings that rank all relevant documents above irrelevant ones require exponential size, whereas polynomial-size multi-vector embeddings suffice. Our result establishes an exponential separation between the expressive power of single-vector and multi-vector embeddings for the task of ranking of documents as opposed to approximating numerical scores as in the work of Jayaram.
 
Motivated by our theoretical construction, we introduce ANDOR, a new retrieval benchmark that naturally instantiates these hard examples. We show that state-of-the-art single-vector embedding models perform poorly on ANDOR in the zero-shot setting and exhibit only marginal improvements after fine-tuning, highlighting the inherent difficulty of the benchmark compared to prior work. In contrast, multi-vector models consistently outperform their single-vector counterparts and improve substantially with fine-tuning, closely aligning with our theoretical predictions.

\end{abstract}
\section{Introduction}

Late-interaction models such as ColBERT~\cite{kz20} have established multi-vector retrieval as a powerful alternative to single-vector dense embeddings. Unlike single-vector models, which represent each query and document with a single embedding and compute their similarity using a dot product, late-interaction models retain token-level embeddings, representing queries and documents as sets of vectors. The relevance of a document to a query is then computed using the Chamfer score, which matches each query vector to its most similar document vector and sums the resulting maximum inner products. The fine-grained token level interaction leads to substantially improved retrieval quality, particularly on complex tasks such as agentic and multimodal retrieval, where late-interaction models currently lead benchmarks such as BrowseComp+ and Vidore v3. Their strong empirical performance has also driven widespread industrial adoption, with systems such as Pinecone~\cite{pinecone_cascading_retrieval} and Mixedbread~\cite{mixedbread_multimodal}, as well as native multi-vector support in vector databases including Vespa~\cite{vespa_multi_vector_hnsw}, Qdrant~\cite{qdrant_late_interaction}, and LanceDB~\cite{lancedb_late_interaction}.

\begin{figure}[t]
    \centering
    \includegraphics[width=0.7\columnwidth]{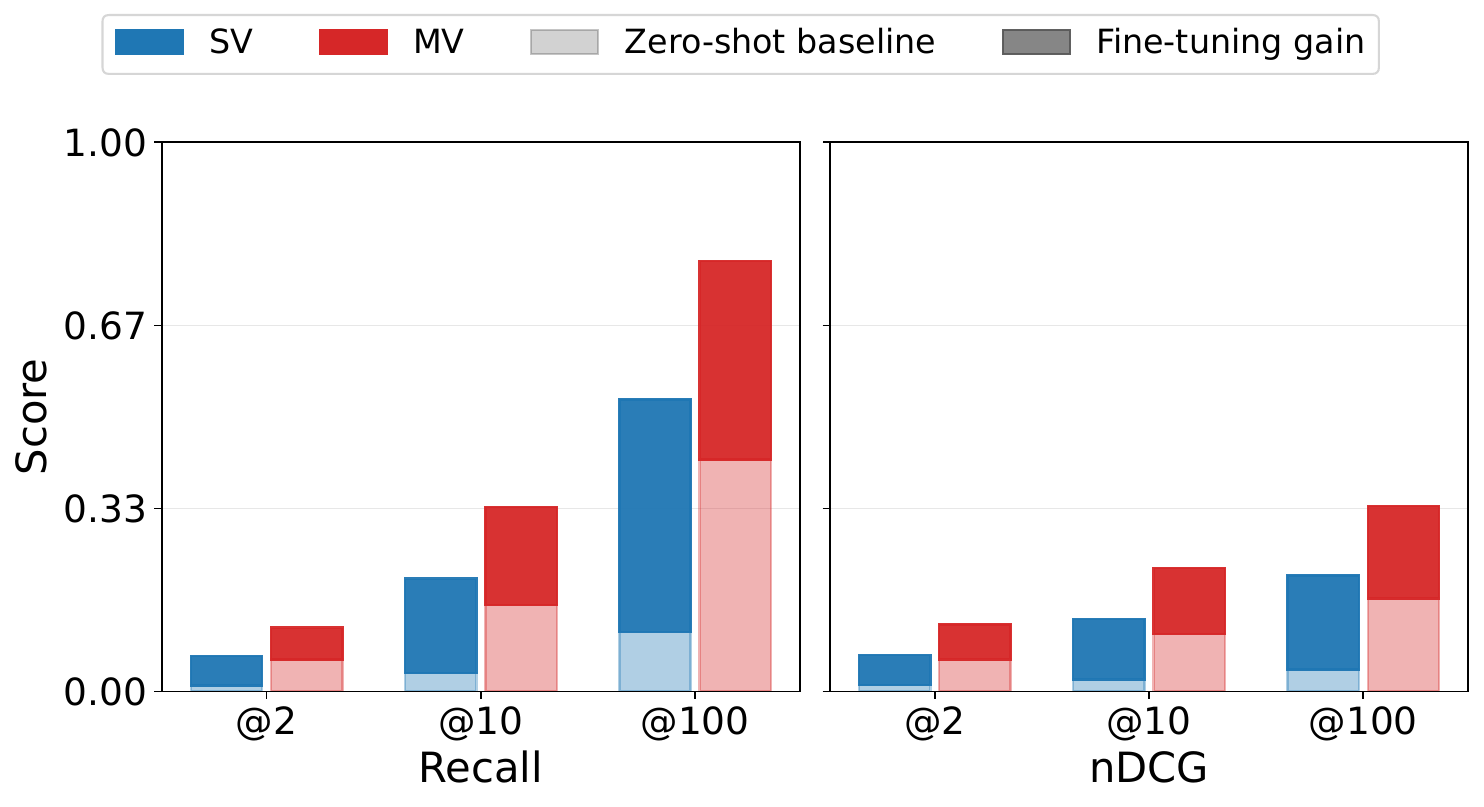}
    \caption{Recall (left) and nDCG (right) for the single-vector (SV; blue) and multi-vector (MV; red) representations of Jina Embeddings v4 after joint fine-tuning on ANDOR at a mean training width of $5.5$, evaluated at test width $3.5$.}
    \label{fig:mainfig}
\end{figure}

Despite their empirical success, it remains unclear whether \emph{multi-vector representations are fundamentally more expressive than single-vector embeddings}. Existing theoretical results only partially address this question. For instance, sign-rank characterizes the dimension required for single-vector embeddings to preserve \emph{retrieval orderings} -- that is, to rank every relevant document above every irrelevant document -- but it says nothing about the representational power of multi-vectors. On the other hand, recent work establishes an exponential separation for preserving Chamfer (MaxSim) similarity scores~\cite{dhr+24,j26}. However, preserving Chamfer similarity scores is a considerably stronger requirement than preserving document orderings. Empirical evidence is similarly inconclusive. The LIMIT benchmark~\cite{wbnl26} identifies retrieval tasks on which state-of-the-art single-vector models perform substantially worse than multi-vector models. However, subsequent work~\cite{sagks26} shows that much of this gap can be recovered through task-specific fine-tuning, raising the possibility that the observed advantage stems from training rather than an inherent representational limitation.

\paragraph{Contributions:} In our work, we resolve the fundamental question of whether multi-vector representations are fundamentally more expressive than single-vector embeddings for retrieval. Specifically, we prove an explicit exponential separation between the two for preserving retrieval orderings. We construct a family of relevance matrices that admits polynomial-size\footnote{The size of a multi-vector representation is defined as the number of vectors per document (or query) multiplied by the dimension of each vector.} multi-vector representations, yet requires exponential-dimensional single-vector embeddings to rank every relevant document above every irrelevant document. Our lower bound builds on sign-rank techniques from~\cite{rs10} and is complementary to that of ~\cite{j26}. While the construction of ~\cite{j26} requires exponentially large single-vector embeddings to preserve Chamfer scores, we show that it nevertheless admits polynomial-dimensional single-vector embeddings that preserve retrieval orderings.

Finally, our theoretical construction naturally yields a realistic benchmark. We introduce ANDOR, a shopping retrieval benchmark whose queries consist of conjunctions of disjunctions, a compositional structure common in faceted search~\cite{BGH+08,tunkelang2009faceted,VFK13}. Consistent with our theoretical results, ANDOR exhibits a persistent separation between single-vector and multi-vector retrieval (See Figure~\ref{fig:mainfig}): unlike LIMIT \cite{wbnl26}, the performance gap remains even after task-specific fine-tuning. Multi-vector models outperform single-vector models from 2-10× in the zero-shot setting and 2× after fine-tuning, while achieving performance comparable to established retrieval benchmarks. We present ANDOR as a challenging benchmark for future retrieval research.

\section{Related Work}

%\subsubsection{Dense and late-interaction retrieval}
Modern retrieval systems have popularised the use of neural embedding models. Dense  Retrieval~\cite{koml+20}, in which a query and document are each mapped to a single-vector embedding and relevance is scored using the cosine similarity. Contemporary
single-vector encoders such as Qwen3 Embedding~\cite{zlz+25} and
Arctic-Embed~\cite{ymy+24} follow this template, typically using contrastive
objectives~\cite{ovk18,sohn16}. Late-interaction models instead have token-level embeddings.  ColBERT~\cite{kz20} scores queries and documents using Chamfer score; ColBERTv2~\cite{skfpz22} improves effectiveness while compressing token embeddings.  MUVERA~\cite{dhr+24} connects the two
regimes by mapping multi-vector embeddings to fixed-dimensional embeddings whose inner
products approximate Chamfer similarity. Models such as Jina Embeddings v4~\cite{gsm+25} are jointly trained to expose both single-vector and multi-vector embedding heads from a single backbone. These models describe practical value of late-interaction models, but do not demonstrate representational differences between the two schemes. 

%\paragraph{Empirical limits and compositional retrieval.}
A second line of work investigates when single-vector retrieval fails. \cite{wbnl26} give a worst case analysis of the dimension needed for single-vector embeddings. Additionally, they present LIMIT, a controlled benchmark on which strong single-vector
models perform poorly. Their theory does not establish a formal single-vector
versus multi-vector separation, nor does it imply that the finite LIMIT relevance
matrix itself requires a large representation dimension. ~\cite{sagks26} further show through controlled experiments that
dimensionality alone does not fully explain the observed LIMIT behavior.

The mathematical machinery used in our separation originates in communication
complexity. Pattern matrices and their sign-rank lower bounds
\cite{rs10,Sherstov11,sw19} provide tools for bounding the dimension of
inner-product realizations of Boolean matrices. In retrieval, ~\cite{j26} show that approximating certain Chamfer similarity matrices by ordinary inner products can require very large single-vector representations.
Their hard instance is constructed from a \NAND pattern matrix with singleton queries and multi-vector documents, and their objective is pointwise approximation of numerical similarity scores. Concurrently, ~\cite{kizm26} study the capacity of late-interaction scoring and variants designed to increase it. The gap between MUVERA’s upper bound~\cite{dhr+24} and Jayaram’s lower bound~\cite{j26} has since been nearly closed by recent work~\cite{jlmw26}.
 
Our objective is structurally different: we require only that every relevant document be ranked above every irrelevant document. Score-approximation hardness does not automatically imply hardness for this ordering objective. Our dataset ANDOR is also complementary to existing benchmarks. Its queries are conjunctions of
disjunctions over product attributes, a compositional structure arising naturally
in multifaceted search in e-commerce and other avenues~\cite{BGH+08,tunkelang2009faceted,VFK13}.

\section{Preliminaries} \label{sec:preliminaries}

Here we formalize the single-vector and multi-vector representations and the notion of retrieval ordering problem studied in this work.
\\
Let $R\in\{0,1\}^{|\cQ|\times|\cD|}$ be a relevance matrix, where $\mathcal Q$ is a set of queries and $\mathcal D$ is a set of documents. If document $j$ is relevant to query $i$, $R_{i,j}=1$ otherwise $R_{i,j}=0$. We want to represent the queries and documents using vector representations, along with a scoring metric, to determine if a document is relevant to a query. In this paper, we focus on the following two popular methods.
\paragraph{Single-vector representation}
A single-vector representation assigns a $d$-dimensional vector $q_i\in\mathbb R^d$ for each query $i$ and a $d$-dimensional vector $p_j\in\mathbb R^d$ for each document $j$. We use inner product, denoted by $\langle q_i,p_j\rangle$ to get the similarity score.

\paragraph{Multi-vector representation}
A multi-vector representation instead assigns a set of vectors $Q(i) = \{q_1,\dots,q_{L_Q}\}$ to each query $i$ and $P(j) = \{p_1,\dots,p_{L_P}\}$ to each document $j$. We use the Chamfer score which for $i$-th query and $j$-th document is defined as:
\begin{align}
S(Q(i),P(j))=\frac{1}{L_Q}\sum_{q\in Q(i)}\max_{p\in P(j)}\langle q,p\rangle. \label{eq:Chamfer}
\end{align}

Our objective is to understand the expressive power of these representations in representing a given relevance matrix. We want that for each query $q$, the scores of the relevant documents be higher than those of irrelevant documents. Formally, the goal of \textbf{retrieval ordering problem} is to design vector representations along with their scoring metric such that,
\begin{equation}
    \min_{p:R(q,p)=1} \text{score}(q,p) > \max_{p:R(q,p)=0} \text{score}(q,p) \label{eq:retrievalordering}
\end{equation}
We therefore seek relevance matrices that admit compact multi-vector representations but require large single-vector representations for preserving the retrieval ordering. We will now introduce pattern matrices as defined in Definition 2.5 by ~\cite{rs10} that will help us to construct the relevance matrix.

Let $x \in \{0,1\}^N$, $f: \{0,1\}^n \to \{0,1\}$, and let $N = c n$ for an integer $c \ge 2$. Partition $[N]$
into $n$ blocks
$$
  B_1, \dots, B_n, \qquad |B_i| = c.
$$
A selector is a vector $\mu\in\{0,1\}^N$ satisfying 
\begin{equation}
    \sum_{t\in B_s}\mu_t=1 \hspace{1cm}\text{for every }s\in[n].
\end{equation}
Thus $\mu$ selects exactly one coordinate of $x$ from every block. Define the extracted string $\pi_\mu(x)\in\{0,1\}^n$ coordinate-wise by
\begin{equation}
    \pi_\mu(x)_s = \sum_{t\in B_s}\mu_t x_t, \hspace{1cm} s\in[n].
\end{equation}
Because exactly one $\mu_t$ in each block is nonzero, this sum is a single bit. For $w\in\{0,1\}^n$, the pattern matrix of $f$ is the boolean matrix
\begin{align}
  \PM(N,n,f)_{(\mu,w),x} = f(\pi_\mu(x) \oplus w)
\end{align}
which has $(2c)^n$ rows and $2^N$ columns, where each row is denoted by the pair $(\mu,w)$ and each column is denoted by $x$.
\\
The tuple $(N,n,f)$ gives a unique pattern matrix belonging to this family of matrices. In our work, we use two functions $f$, namely, Minsky-Papert function, which is popular in circuit complexity and $\nand_n$ function.

\paragraph{Minsky-Papert functions}
For positive integers $k_1, k_2$, define\\$\MP_{k_1,k_2}(z): \{0,1\}^{k_1k_2} \to \{0,1\}$ as:
\begin{equation}
    \MP_{k_1,k_2}(z) = \bigwedge_{i=1}^{k_1} \bigvee_{j=1}^{k_2} z_{(i-1)k_2+j}, \qquad z \in \{0,1\}^{k_1k_2}.
\end{equation}

A straightforward observation is that if there is a single-vector representation in dimension $d$ preserving retrieval ordering of relevance matrix $R$, then it implies that $d+1 \ge \sr{R}. \label{eq:signrankbasic}$ For a boolean matrix $M$, we define $\sr{M}=\sr{M^{\pm}}$, where $M^{\pm}_{ij}=(-1)^{M_{ij}+1}$.
Consequently, lower bounds on sign-rank translate directly into lower bounds on the embedding dimension of single-vector representation. We use sign-rank lower bounds proven by ~\cite{rs10} to argue about our relevance matrices:
\begin{lemma}\label{lem:razborov}{\cite[Section~6, proof of Theorem~1.1]{rs10}}
For $n=4m^3$ and $N=17^6n$. The $(N,n,\MP_{m,4m^2})$-pattern matrix
$$ M_{(\mu,w),x}=\MP_{m,4m^2}(\pi_\mu(x) \oplus w)$$
satisfies $\sr{M} = 2^{\Omega(m)}.$
\end{lemma}

\section{Overview of Results}
\label{sec:overview}

Here we state all our main results. In our first theoretical result, we show an exponential separation between the single-vector and multi-vector representations for the retrieval ordering problem.
\begin{restatable}{theorem}{MainTheorem}
     For any integer $m \ge 2$, let $L=4m^2$, $n=mL$ and $N=17^6n$. Then there exists a Boolean relevance matrix $R = \PM(N,n,\MP_{m,L})$ such that any unit-norm single-vector representation which preserves retrieval ordering of $R$ (Eq.~\ref{eq:retrievalordering}) requires an embedding dimension of $d = 2^{\Omega(m)}$. In contrast, there exists a multi-vector representation which preserves retrieval ordering of $R$, using the Chamfer score, with a polynomial representation size of $\mathcal{O}(m^6)$ per query and document, achieving a relevance separation margin of $\Theta(m^{-2})$ between relevant and irrelevant documents. \label{thm:main}
\end{restatable}
The lower bound for the single-vector embeddings follows immediately from \cite{wbnl26} and Lemma \ref{lem:razborov}. For completeness, the proof of it is provided in Section~\ref{app:singlelowerbound}. The construction for the multi-vector embedding is provided in Section~\ref{sec:mainresult}. Note that the above theorem does not refute the MUVERA ~\cite{dhr+24} result which provides a single-vector embedding that approximates the multi-vector chamfer similarity scores up to additive $\epsilon$ approximation with $\exp(O(1/\epsilon^2))$ dimensional single-vector embeddings. For these embeddings to preserve the retrieval ordering one would need to choose $\epsilon = \Theta(m^{-2})$. 

Additionally, we remark here that our result is complementary to the result of ~\cite{j26} which provides a lower bound on the dimension of single-vector embedding needed to approximate the multi-vector Chamfer scores. 
In fact, for the relevance matrix constructed by ~\cite{j26}, we further construct a single-vector representation of size $\Theta(N)$, which matches, up to polynomial factors, the $O(N^2)$ size multi-vector representation given in their work.

\begin{restatable}{theorem}{NandTheorem}
The relevance matrix $R = \PM(N,n,\nand_n)$ from ~\cite{j26} admits a unit-norm single-vector representation which preserves retrieval ordering of $R$ in dimension $N$ with a relevance separation margin of $2/(\sqrt{nN})$ between relevant and irrelevant documents. 
\label{thm:NAND-failure}
\end{restatable}
Thus, while the hard instance in their work successfully demonstrates an exponential gap in pointwise Chamfer score preservation, it does not answer the questions related to retrieval ordering problem.

Motivated by our theoretical findings, we introduce ANDOR, a retrieval benchmark designed to expose the expressiveness gap between single-vector and multi-vector retrieval. As described in Section~\ref{sec:dataset}, we translate the Boolean retrieval problem used in our theoretical construction into an e-commerce retrieval task by representing variables as product attributes and clauses as user preferences. The dataset is meant to emulate a complicated multi-faceted search task, that appears often in retail. The vocabulary consists of 20 product categories, each containing 20 distinct attributes, and the benchmark evaluates a model's ability to capture compositional AND-of-OR semantics, where a relevant document must satisfy multiple mandatory attribute constraints specified in the query. The corpus contains 50,000 carefully constructed products together with challenging hard negatives that violate only a small number of constraints, preventing models from succeeding through partial matches or simple heuristics. Finally, we vary the query width to systematically study how retrieval performance changes as the underlying Boolean reasoning becomes more complex. Further details are provided in Section~\ref{sec:dataset}.

In Section~\ref{sec:experiments} and Appendix~\ref{sec:app-empirical}, we summarize the results of extensive experiments on ANDOR using wide range of models and varying ablation settings. We observe that while zero-shot performance is consistently poor, multi-vector models maintain a significant lead. Fine-tuning helps all models, but multi-vector models maintain a significant lead in all cases; post fine-tuning relative margin ranges from up to 153\% for Recall@2 to up to 96\% for Recall@100. Most strikingly, we note that even in case of the Jina v4 model, where multi and single-vector embeddings are co-trained under identical circumstances, the gap remains, pointing to a fundamental gap in representation and not merely an artifact of training or hyperparameter settings.
\section{Proofs}\label{sec:mainresult}
In this section, we focus on constructing a polynomial-size multi-vector representation, and provide a proof of Theorem~\ref{thm:main}. At the end, we will show via construction that the $\nand_n$ pattern matrix from ~\cite{j26} admits a unit norm single-vector representation that matches, up to polynomial factors, the multi-vector that they show in their work. 

We begin with the setup first. Our query will be a composition of several clauses. $m$ is the number of clauses, or independent conditions, that a document must satisfy, and $L$ is the number of ways each clause can be satisfied; a document will be relevant to a query only iff \emph{every} one of the $m$ clauses is satisfied by at least one of its $L$ literals. The quantity $n=mL$ is the total number of logical bit positions needed to describe one instance of this clause structure, arranged as an $m\times L$ array. We will refer to the array and its rows and columns in later steps in the construction. \\
We denote a document by $x$ and a query by the pair $(\mu,w)$ as defined earlier in Section~\ref{sec:preliminaries}. Let $y \in \{0,1\}^n$ be defined as the following:
\begin{align}
 y&=\pi_\mu(x)\oplus w
\end{align}
We use $R=\PM(N,n,\MP_{m,L})$ as our relevance matrix:
\begin{equation}
 R_{(\mu, w),x} = \MP_{m,L}(y). \label{eq:MPdef}
\end{equation}
where $m \ge 2$ is any integer and $L=4m^2$, $n=mL=4m^3$ and $N=17^6n=4 \cdot 17^6m^3$ for $c=17^6$.
\\ \\
In the remainder of this section, we focus on the multi-vector construction. Specifically, we construct an exact multi-vector realization of the same relevance matrix and show that the resulting multi-vector representation has polynomial size. Let us first arrange $y = \pi_\mu(x) \oplus w$ in a $m \times L$ matrix:
\begin{equation}
\begin{pmatrix}
y_1 & y_2 & \cdots & y_L\\
y_{L+1} & y_{L+2} & \cdots & y_{2L}\\
\vdots & \vdots & \ddots & \vdots\\
y_{(m-1)L+1} & y_{(m-1)L+2} & \cdots & y_{mL}
\end{pmatrix} 
\in\{0,1\}^{m\times L}. \label{eq:f-matrix}
\end{equation}
We will refer to this boolean matrix as $Y(\mu,w,x)$, and use $Y$ to distinguish the matrix from the function that follows the same structure. The $\MP_{m,L}(y)$ takes the OR of each row of $Y$ and then finally take the AND of the obtained $m$ values. Thus, the function returns $1$ if and only if each row of $Y$ contains at least one $1$. So each element of $Y$ is a literal and each row corresponds to a clause. For instance, if we say clause $i$ is satisfied, it means the $i$-th row of $Y$ contains at least a one $1$. We will work in $\mathbb R^{N}$, endowed with the standard basis $\{e_1,\ldots, e_N\}$ and inner product. For every query $(\mu,w)$, define a collection of $m$ vectors 
\begin{equation}
    Q(\mu,w) = \{q_i: i \in [m]\}, 
\end{equation}
where the vector $q_i$ is associated with row $i$ of $Y$. Intuitively, it can be imagined as encoding the $i$-th clause, with a sign bit to encode the bit value, for every coordinate $\mu$ chooses in the document for these $L$ literals in the clause and normalized to be unit norm. The sign bit is masked using $w$.
\begin{equation}
    q_i = \frac{1}{\sqrt{L}} \sum_{s=(i-1)L+1}^{iL} \left((-1)^{w_s} \sum_{t \in B_s}\mu_te_t\right). \label{eq:mvq}
\end{equation}
For every document $x$, define one vector for each bit $i$ :
\begin{equation}
    P(x) = \{p_j: j \in [N]\}. \label{eq:mvd}
\end{equation}
\begin{equation}
    p_j = (-1)^{x_j+1}e_j.
\end{equation}

In the following paragraphs, we first compute the relevant inner products and then show that the maximum associated with clause $i$ is $1/\sqrt{L}$ when that clause is satisfied and $0$ otherwise. Consequently, the Chamfer score is $1/\sqrt{L}$ when all clauses are satisfied and at most $(m-1)/(m\sqrt{L})$ when at least one clause is unsatisfied. This yields a gap between relevant and irrelevant pairs and allows us to choose a separating threshold. First, an easy lemma.

\begin{lemma}
For $s \in [n]$, let $v_s$ be the unique $t \in B_s$ satisfying $\mu_t=1$. For every $i\in[m]$, let $I_i = \{v_s: (i-1)L+1 \le s \le iL\}$. Then,
\begin{equation}
\langle q_i,p_r\rangle
=
\begin{cases}
+1/\sqrt{L}, & r=v_s \in I_i,y_{s}=1,\\
-1/\sqrt{L}, & r=v_s \in I_i,y_{s}=0,\\
0, & r \notin I_i,
\end{cases}
\end{equation}

\label{lem:innerproduct}
\end{lemma}
\begin{proof}
    We consider two cases whether $r$ is in $I_i$ or not.\\
    \textbf{Case 1:} $r=v_s \in I_i$\\
    In this case, the non zero entry of the document vector $p_r$ is at the index $r$ and since $\mu_r=1$, the index $r$ is nonzero for $q_i$. It therefore follows that
    \begin{equation}
    \begin{split}
        \langle q_i,p_r\rangle        
        &=\frac{1}{\sqrt{L}}(-1)^{w_s}(-1)^{x_{v_s}+1}\\
        \langle q_i,p_{v_s}\rangle &=\frac{1}{\sqrt{L}}(-1)^{(w_s \oplus x_{v_s})+1}=\frac{1}{\sqrt{L}}(-1)^{y_s+1}
    \end{split}
    \end{equation}
    \begin{equation}
    \text{Thus, } \langle q_i,p_{v_s}\rangle =
    \begin{cases}
    +1/\sqrt{L}, & y_{s}=1,\\
    -1/\sqrt{L}, & y_{s}=0,
    \end{cases}
    \end{equation}
    \textbf{Case 2:} $r \notin I_i$\\
    In this case, the index $r$ is zero in $q_i$ and therefore $\langle q_i,p_r\rangle=0$. 
\end{proof}
Consequently, as $m\geq2$,

\begin{equation}
\max_{p\in P(x)}\langle q_i,p\rangle =
\begin{cases}
1/\sqrt{L}, & \hfill\textrm{if } \displaystyle\bigvee_{j=1}^{L}y_{(i-1)L+j}=1,\\
0, & \hfill\textrm{if } \displaystyle\bigvee_{j=1}^{L}y_{(i-1)L+j}=0.
\end{cases}
\end{equation}
We now have enough to describe the behavior of the Chamfer score.
\begin{lemma}
    The Chamfer score satisfies
    \begin{equation}
        S(Q(\mu,w),P(x)) 
        \begin{cases}
        = \dfrac{1}{\sqrt{L}} \hspace{1cm}\textrm{if } \MP_{m,L}(y) =1 \\
        \leq \dfrac{m - 1}{m\sqrt{L}}  \hspace{1cm} \textrm{o.w. } \\
    \end{cases} 
    \label{lem:Chamfer}
    \end{equation}\label{lem:Chamfer}
\end{lemma}
\begin{proof}
    The proof follows trivially from Lemma~\ref{lem:innerproduct}: if $\MP_{m,L}(y) = 1$, every row of $Y$ contains at least one $1$. Each row contributes exactly $1/\sqrt{L}$ to the Chamfer score, giving the first case. The other case implies at least some row yields a $0$, which means that the Chamfer score can at most be $1/\sqrt{L} - 1/m\sqrt{L}$, which resolves to the other value after algebraic rearrangement.
\end{proof}
A solution is to pick a threshold $\gamma$ that sits in between these two values. Choosing $\gamma = (m-1/2)/(m\sqrt{L})$, we obtain
\begin{equation}
    S(Q(\mu,w),P(x))>\gamma \Longleftrightarrow \MP_{m,L}(\pi_\mu(x)\oplus w)=1,
\end{equation}
implying an exact reconstruction of the relevance matrix $R$. To prove the main theorem, we still need a lower bound on the size of the single vectors for the retrieval ordering problem.

\subsection{A Lower Bound for Single Vectors} \label{app:singlelowerbound}
\begin{restatable}{lemma}{LowerBoundTheorem} \label{lem:singlevectorlowerbound}
    Let $R$ be the relevance matrix given as $\PM(N,n,\MP_{m,4m^2})$ with $n=4m^3$ and $N=17^6n$ as defined above. Then for any single-vector representation in dimension $d$ which preserves retrieval ordering of $R$, it implies that $d = 2^{\Omega(m)}$.
\end{restatable}
\begin{proof}
    Let a query $(\mu,w)$ and a document $x$ be encoded by embeddings $a_{\mu,w}, b_x \in \mathbb R^d$, where $R$ is the relevance matrix. Let the query and document matrices be denoted therefore as $A$ and $B$ respectively, and the inner product matrix as $G= AB^\top$. Clearly, rank$(G)\leq d$. Assume a (possibly query dependent) threshold $\gamma_{\mu,w}$ such that:
    \begin{equation}
        R_{(\mu,w),x} = 1 \Longleftrightarrow G_{(\mu,w),x} \coloneq \langle a_{\mu,w}, b_x \rangle > \gamma_{\mu,w}.
    \end{equation}
    We must construct a signrank witness out of $G$, and therefore the witness cannot have zero entries. To ensure this, we define for each query $(\mu,w)$, a small, finite perturbation $\eps_{\mu,w} > 0$ such that no relevant query document pair $((\mu,w),x^+)$ satisfies $G_{(\mu,w),x^+}\leq \gamma_{\mu,w} +\eps_{\mu,w}$, where $x^+$ denotes positive documents. Let $\tilde \gamma_{\mu,w} \coloneqq \gamma_{\mu,w} + \eps_{\mu,w}$, and $\tilde \gamma$ the vector composed of them. Define the witness
    \begin{equation}
        T \coloneqq G -\tilde \gamma\mathbf1^\top.
    \end{equation}
    The entries of $T$ therefore exactly satisfy the condition that they are positive for relevant pairs and negative for irrelevant pairs, and thus is a valid witness to the signrank of $R$. The second term has rank at most $1$, following from the rank of the all ones matrix. Therefore, it follows that $\rank (T)\leq \rank(G) + 1 \leq d + 1.$ From Lemma~\ref{lem:razborov}, it implies that $\signrank(R) = 2^{\Omega(m)}$. Hence, $d \geq \signrank(R) - 1 = 2^{\Omega(m)}.$
\end{proof}

\subsection{Proof of Theorem~\ref{thm:main}}
The proof of the main theorem follows from the aforementioned lemmas. For completeness, we restate the theorem here.
\MainTheorem*
\begin{proof}
Recall that our relevance matrix is the one defined in Eqn~\ref{eq:MPdef}. We spell out the proof below.
    \paragraph{Space.} The lower bound on the size of the single-vector representation follows directly from Lemma~\ref{lem:singlevectorlowerbound}. For the multi-vector construction described earlier in the section, we produce vectors in $N$ dimensional space: a query is represented by a set of $m$ vectors, and a document by a set of $N = 4cm^3$ vectors. Therefore, the representation size is $O(N^2)=O(m^6)$ and thus polynomial in $m$, in sharp contrast to the exponential lower bound on single vectors. 
    \paragraph{Separation. } It remains to show what gap our representation achieves in between relevant and irrelevant documents. This follows immediately from our proof in Lemma~\ref{lem:Chamfer}; relevant and irrelevant documents must have a gap of at least $1/m\sqrt{L}$. Plugging in the values for $L$, it follows that the gap is $1/2m^2 = \Theta(m^{-2}).$
\end{proof}
\subsection{The $\nand_n$ Pattern Matrix admits a linear single-vector representation}\label{app:nand}
As a final theoretical result, we show that the $\nand_n$ pattern matrix defined in ~\cite{j26} and the associated multi-vector construction does not solve the retrieval-ordering problem we consider in this work.
\NandTheorem*
We now construct a single-vector representation for $R_{(\mu,w),x}=\PM(N,n,\nand_n)$ and prove Theorem~\ref{thm:NAND-failure}. We retain our convention that $(\mu,w)$ indexes queries and $x$ indexes documents. Our vectors lie in $\mathbb R^{N}$. For query $(\mu,w)$, define 

\begin{equation}
q_{\mu,w} = \dfrac{1}{\sqrt n} \sum_{s=1}^{n}\sum_{t\in B_s}\mu_t (-1)^{w_s} e_t.
\end{equation}
For each block $s$, the selector chooses one coordinate $t$, and the query vector uses $e_t$ when $w_s=0$ and $-e_t$ when $w_s=1$. There are $n$ non-zero entries in query vector so we divide by $\sqrt{n}$ to normalize it to a unit vector. Similarly, pool the documents as

\begin{equation}
p_{x} = \dfrac{1}{\sqrt N}\sum_{t=1}^{N}(-1)^{x_t}e_t.
\end{equation}
For each position $t$, the coordinate $e_t$ records the bit $x_t$: the summand uses $e_t$ when $x_t=0$ and $-e_t$ when $x_t=1$. We divide by $\sqrt N$ to normalize it to a unit vector.
\\ \\
Consequently, contribution of index $t$ to the inner product $\langle q_{\mu,w}, p_x \rangle$ is $1$ exactly when the selected bit $x_t$ agrees with $w_s$ and $-1$ if the selected bit $x_t$ disagrees with $w_s$. Contribution of $t$ to the inner product is $0$ for unselected bit $x_t$. Let
\begin{equation}
r = | \{s\in[n]:\pi_\mu(x)_s=w_s\}|
\end{equation}
be the number of selected coordinates on which $x$ agrees with the mask $w$. By construction,

\begin{equation}
\langle q_{\mu,w},p_x\rangle = \dfrac{2r-n}{\sqrt{nN}}.
\end{equation}
The $\nand$ value is zero exactly when every selected bit disagrees with the mask, equivalently when $r=0$; it is one exactly when $r\geq1$. Thus,
\begin{equation}
\langle q_{\mu,w},p_x\rangle
\begin{cases}
\ge (2-n)/\sqrt{nN}, & \nand_n(\pi_\mu(x) \oplus w)=1,\\
= -n/\sqrt{nN}, & \nand_n(\pi_\mu(x) \oplus w)=0.
\end{cases}
\label{eq:nandseparator}
\end{equation}
Therefore, by setting a common threshold

\begin{equation}
\gamma = \dfrac{1-n}{\sqrt{nN}},
\end{equation}
we have
\begin{align}
\langle q_{\mu,w},p_x\rangle>\gamma \hfill \Longleftrightarrow \hfill \nand_n(\pi_\mu(x) \oplus w)=1.
\end{align}
We now provide proof of Theorem~\ref{thm:NAND-failure}.

\begin{proof}[Proof of Theorem~\ref{thm:NAND-failure}]
From Eq.~\ref{eq:nandseparator}, an irrelevant pair has score $-\sqrt{n/N}$, whereas the smallest relevant score is $(2-n)/\sqrt{nN}$. Thus the gap between relevant and irrelevant document is at least $2/\sqrt{nN}$. The representation preserves the threshold exactly like in our case and has a relevance separation margin of $2/\sqrt{nN}$ which in the notation used in ~\cite{j26} translates to $2/\sqrt{mk}$. The size of our single-vector representation is $\Theta(N)$ and the multi-vector representation used in the work of ~\cite{j26} uses one vector per query and $N$ vectors per document with each vector being in dimension $N$ and hence the size of their multi-vector representation is $O(N^2)$. Thus, our single-vector representation matches up to polynomial factors the size of their multi-vector representation.
\end{proof}
\section{The ANDOR dataset}\label{sec:dataset}
\begin{figure*}[t]
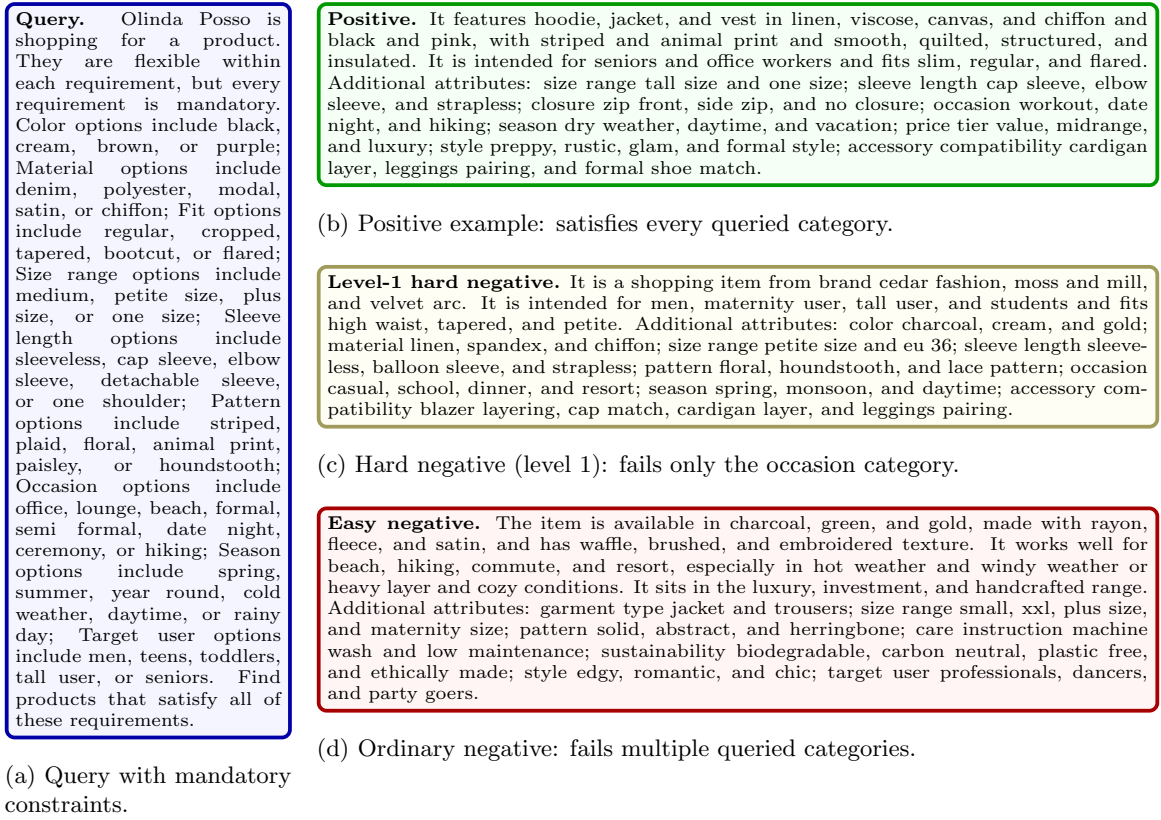

\centering
\captionsetup[subfigure]{font=footnotesize,justification=raggedright,singlelinecheck=false,skip=2pt}
\captionsetup{font=small}

\newlength{\QueryW}
\setlength{\QueryW}{0.23\textwidth}
\newlength{\MidGap}
\setlength{\MidGap}{0.02\textwidth}
\newlength{\StackW}
\setlength{\StackW}{0.9\dimexpr\textwidth-\QueryW-\MidGap\relax}

\begin{tabular}{@{}p{\QueryW}@{\hspace{\MidGap}}p{\StackW}@{}}
\vspace{0pt}
\subcaptionbox{Query with mandatory constraints.}[ \linewidth ]{%
\begin{tcolorbox}[exbox,colback=blue!4,colframe=blue!65!black]
\fontsize{6.8}{8.0}\selectfont
\textbf{Query.} Olinda Posso is shopping for a product. They are flexible within each requirement, but every requirement is mandatory. Color options include black, cream, brown, or purple; Material options include denim, polyester, modal, satin, or chiffon; Fit options include regular, cropped, tapered, bootcut, or flared; Size range options include medium, petite size, plus size, or one size; Sleeve length options include sleeveless, cap sleeve, elbow sleeve, detachable sleeve, or one shoulder; Pattern options include striped, plaid, floral, animal print, paisley, or houndstooth; Occasion options include office, lounge, beach, formal, semi formal, date night, ceremony, or hiking; Season options include spring, summer, year round, cold weather, daytime, or rainy day; Target user options include men, teens, toddlers, tall user, or seniors. Find products that satisfy all of these requirements.
\end{tcolorbox}
}%
&
\vspace{0pt}
\parbox[t]{\linewidth}{%
\subcaptionbox{Positive example: satisfies every queried category.}[ \linewidth ]{%
\begin{tcolorbox}[exbox,colback=green!4,colframe=green!60!black]
\fontsize{6.8}{8.0}\selectfont
\textbf{Positive.} It features hoodie, jacket, and vest in linen, viscose, canvas, and chiffon and black and pink, with striped and animal print and smooth, quilted, structured, and insulated. It is intended for seniors and office workers and fits slim, regular, and flared. Additional attributes: size range tall size and one size; sleeve length cap sleeve, elbow sleeve, and strapless; closure zip front, side zip, and no closure; occasion workout, date night, and hiking; season dry weather, daytime, and vacation; price tier value, midrange, and luxury; style preppy, rustic, glam, and formal style; accessory compatibility cardigan layer, leggings pairing, and formal shoe match.
\end{tcolorbox}
}

\vspace{3mm}

\subcaptionbox{Hard negative (level 1): fails only the occasion category.}[ \linewidth ]{%
\begin{tcolorbox}[exbox,colback=yellow!5,colframe=yellow!55!black]
\fontsize{6.8}{8.0}\selectfont
\textbf{Level-1 hard negative.} It is a shopping item from brand cedar fashion, moss and mill, and velvet arc. It is intended for men, maternity user, tall user, and students and fits high waist, tapered, and petite. Additional attributes: color charcoal, cream, and gold; material linen, spandex, and chiffon; size range petite size and eu 36; sleeve length sleeveless, balloon sleeve, and strapless; pattern floral, houndstooth, and lace pattern; occasion casual, school, dinner, and resort; season spring, monsoon, and daytime; accessory compatibility blazer layering, cap match, cardigan layer, and leggings pairing.
\end{tcolorbox}
}

\vspace{3mm}

\subcaptionbox{Ordinary negative: fails multiple queried categories.}[ \linewidth ]{%
\begin{tcolorbox}[exbox,colback=red!4,colframe=red!65!black]
\fontsize{6.8}{8.0}\selectfont
\textbf{Easy negative.} The item is available in charcoal, green, and gold, made with rayon, fleece, and satin, and has waffle, brushed, and embroidered texture. It works well for beach, hiking, commute, and resort, especially in hot weather and windy weather or heavy layer and cozy conditions. It sits in the luxury, investment, and handcrafted range. Additional attributes: garment type jacket and trousers; size range small, xxl, plus size, and maternity size; pattern solid, abstract, and herringbone; care instruction machine wash and low maintenance; sustainability biodegradable, carbon neutral, plastic free, and ethically made; style edgy, romantic, and chic; target user professionals, dancers, and party goers.
\end{tcolorbox}
}
} \\
\end{tabular}

\caption{Example query with mandatory constraints across nine categories (AND-of-OR semantics) on the left. On the right, we provide one matching document (positive), one near-miss (level-1 hard negative), and one broad mismatch (easy negative).}
\label{fig:worked-scoring-example}
\end{figure*}
We now present the ANDOR dataset, which incorporates these ideas into a hard setting for retrieval. In this setting, we design a dataset meant to emulate shopping queries in e-commerce. The queries are designed to mimic customer preferences in shopping for clothing: we imagine twenty filter groups for items on the inventory, and each filter has twenty possible options. We also refer to these filters as \emph{categories}, and the possible values inside each category as an attribute. Shopping queries naturally specify mandatory categories while allowing multiple acceptable values within each category.

The relevance rule is defined by the Minsky-Papert function
\begin{equation}
    f(y)=\bigwedge_{i=1}^{m}\bigvee_{j=1}^{L} y_{(i-1)L + j}.
\end{equation}

Interpret each variable $y$ as a Boolean indicator that the document satisfies the $j$-th acceptable attribute in category $i$. The ANDOR benchmark is obtained by replacing these abstract Boolean variables with semantic product attributes. Each row $i$ now corresponds to a product category (such as color or material), while the literals within that row correspond to the acceptable attribute values for that category. For a query $q$ and document $d$, define $y_{ij}=1$ iff the document contains the $j$-th accepted attribute for category $i$ requested by the query.

Under this interpretation, $\bigvee_{j=1}^{L} y_{(i-1)L + j}=1$ exactly when the document matches at least one acceptable value in
category $i$, while $\bigwedge_{i=1}^{m}\bigvee_{j=1}^{L} y_{(i-1)L + j}=1$ exactly when every queried category is satisfied. Consequently, the retrieval rule implemented by the benchmark is the same AND-of-OR Boolean function used in our theoretical construction, differing only in the semantic interpretation of the variables. Consider the category \texttt{color}, which can take any of twenty values below:

\begin{tcolorbox}[width=\linewidth, colback=blue!5, colframe=blue!75!black,  halign=left, arc=2mm]
Black, White, Navy, Charcoal, Cream, Beige, Brown, Red,
Burgundy, Pink, Lavender, Purple, Blue, Teal, Green, Olive,
Yellow, Orange, Silver, Gold
\end{tcolorbox}

Then, the query clause corresponding to color achieves $y_{\texttt{color},\texttt{black}} = 1$ for queries that ask for the color \texttt{black}. 
We provide the full list categories and attributes in the appendix~\ref{sec:app-dataset}. The dataset includes 50,000 individual documents. Each document bears the description of a product, each with 11 to 15 categories and 2 to 4 attributes per category. For most of the settings, we maintain closeness to the existing LIMIT dataset~\cite{wbnl26}.

 \subsubsection*{Assigning Relevance.}

Let $C_q$ denote the categories specified by query $q$, $A_q(c)$ the set of acceptable attribute values for category $c$, and $V_d(c)$ the set of values present in document $d$. A document is relevant iff it satisfies every queried category by matching at least
one acceptable value:
\begin{align}
    d \textrm{ is relevant to } q &\iff \forall c \in C_q,\hfill V_d(c) \cap A_q(c) \ne \varnothing,
\end{align}
the exact semantic realization of the Boolean function above where each category corresponds to one disjunctive clause of the conjunction and each acceptable attribute corresponds to one literal within the disjunction. Categories are therefore combined by AND, while accepted attribute values within each category are combined by OR.

We clarify two important parameters for our experiments.
\begin{asparaitem}
    \item \emph{Train Width:} The mean number of attributes in each category of a query used to fine-tune the model.
    \item \emph{Test Width:} The mean number of attributes in each category of a query in the test dataset.
\end{asparaitem}
These are essentially the mean length of a clause in a query in training and testing scenarios. For each, we consider a suite identified by the width. In our experiments, this width is of the form $k+0.5$, where $k$ is an integer and the range of clause lengths in that suite varies from $k-1$ to $k+3$. Each suite in test has an equal number of queries with 7, 8 or 9 categories, and train has an equal number of queries with 5, 6, or 7 categories. For fine-tuning, the training width varies from 5.5 to 9.5; in test, this number goes from 3.5 to 11.5. We identify a query suite with this number. Our dataset contains 1000 test and 800 train queries for each setting. Our testing setting is sparse: each query has exactly two positive documents. For fine-tuning, we adopt a dense training scenario, where we have far more positives per query. The mean number of positives for fine-tuning ranges from $\sim208$ for the 5.5 suite, to $\sim1186$ for the 9.5 suite. Our investigation showed that dense fine-tuning outperforms sparse training. In the generation process, we actively use collision repair to maintain the exact number of positives for test queries. We elaborate on the relevance labels in Appendix ~\ref{sec:app-dataset}.

We consider three levels of hard negatives per query, denoted by the number of categories that they fail to satisfy; level $k$ hard negatives violate exactly k queried categories.
The document corpus contains 2000 designated test positives and 38,780 shared hard negatives across all the 9 test suites as described in Appendix \ref{sec:app-dataset}. The remaining documents are generated as random distractors. Relevance judgments are computed directly from the AND-of-OR rule over the corpus. Refer to Figure~\ref{fig:worked-scoring-example} for an explicit example.

ANDOR therefore provides a natural retrieval benchmark whose relevance labels are generated by the same logical predicate used in our theoretical construction, while remaining interpretable as realistic faceted shopping queries. In the next section, we evaluate performance on this dataset across multiple models and different settings to explore the limits of different models. The dataset and the official code will be released shortly. 
\section{Experiments}\label{sec:experiments}

\begin{table}[t]
\centering
\small
\resizebox{0.7\columnwidth}{!}{%
\begin{tabular}{llccc}
\toprule
Model & Type & Dim. & Open & FT \\
\midrule
GTE ModernColBERT v1~\cite{chaffin25} & Multi-vector  & 128/token & \cmark & \cmark \\
Jina Embeddings v4~\cite{gsm+25}      & Multi-vector  & 128/token & \cmark & \cmark \\
Jina Embeddings v4~\cite{gsm+25}      & Single-vector & 2048      & \cmark & \cmark \\
Qwen3 Embedding 0.6B~\cite{zlz+25}    & Single-vector & 1024      & \cmark & \cmark \\
Snowflake Arctic Embed L v2~\cite{ymy+24} & Single-vector & 1024  & \cmark & \cmark \\
Cohere Embed v4~\cite{cohere_embed4}  & Single-vector & 1536      & \xmark & \xmark \\
OpenAI text-embedding-3-large~\cite{openai_embeddings3} & Single-vector & 3072 & \xmark & \xmark \\
\bottomrule
\end{tabular}
}
\caption{Embedding models evaluated on ANDOR.}
\label{tab:models}
\vspace{-0.5cm}
\end{table}
We evaluate seven recent retrieval models spanning both single-vector and multi-vector retrieval (Table~\ref{tab:models}).
The multi-vector models are GTE ModernColBERT v1 and Jina Embeddings v4 operating in its late-interaction mode. The
single-vector models are Jina Embeddings v4, Qwen3 Embedding 0.6B, Snowflake Arctic Embed L v2, Cohere Embed v4, and
OpenAI text-embedding-3-large. Jina Embeddings v4 jointly produces single-vector and multi-vector representations.
Comparing its two output modes isolates the effect of the representation when the architecture, training data,
and parameter count are identical.

\subsection{Evaluation modes}

\paragraph{Zero-shot:} We first evaluate all the seven models in a zero-shot setting. We rank the single-vector models using cosine similarity and multi-vector models with Chamfer score as defined in Equation~\ref{eq:Chamfer}. Scoring is done against the full corpus of $50{,}000$ documents, and nine suites of $1{,}000$ test queries for various widths (see Section~\ref{sec:dataset} for the definition of width). 

\paragraph{Fine-tuned evaluation:} We fine-tune the five open source models on the $800$ training queries, one run per training width, and evaluate on all the test widths. Each run lasts $6{,}250$ steps at a batch size of $32$ queries; we observe that the models converge in these settings. Every query in a batch is scored against $32$ documents: $2$ positive documents and $30$ negative documents, of which $18$ are hard negatives that satisfy all but one of the queried categories, $6$ satisfy all but two, $3$ satisfy all but three, and $3$ are sampled at random. The two runs trained with the one-positive objective, GTE ModernColBERT and the separately trained Jina multi-vector head, instead score each query against $1$ positive and $31$ negative documents, drawn from a separately constructed pool in which the same $18$, $6$, and $3$ hard-negative tiers are retained and the random tier is increased from $3$ to $4$. 

\begin{figure*}[t]
    \centering
    \begin{subfigure}[t]{\textwidth}
        \centering
        \includegraphics[width=\linewidth]{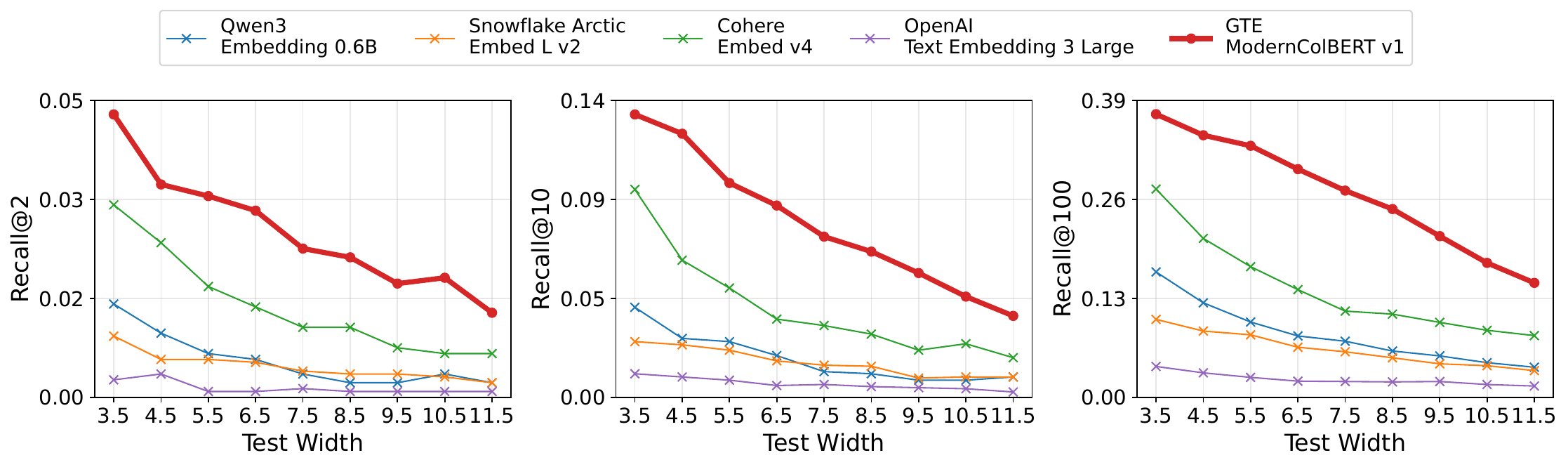}

        \vspace{0.75mm}
    \end{subfigure}
    \begin{subfigure}[t]{\textwidth}
        \centering
        \includegraphics[width=\linewidth]{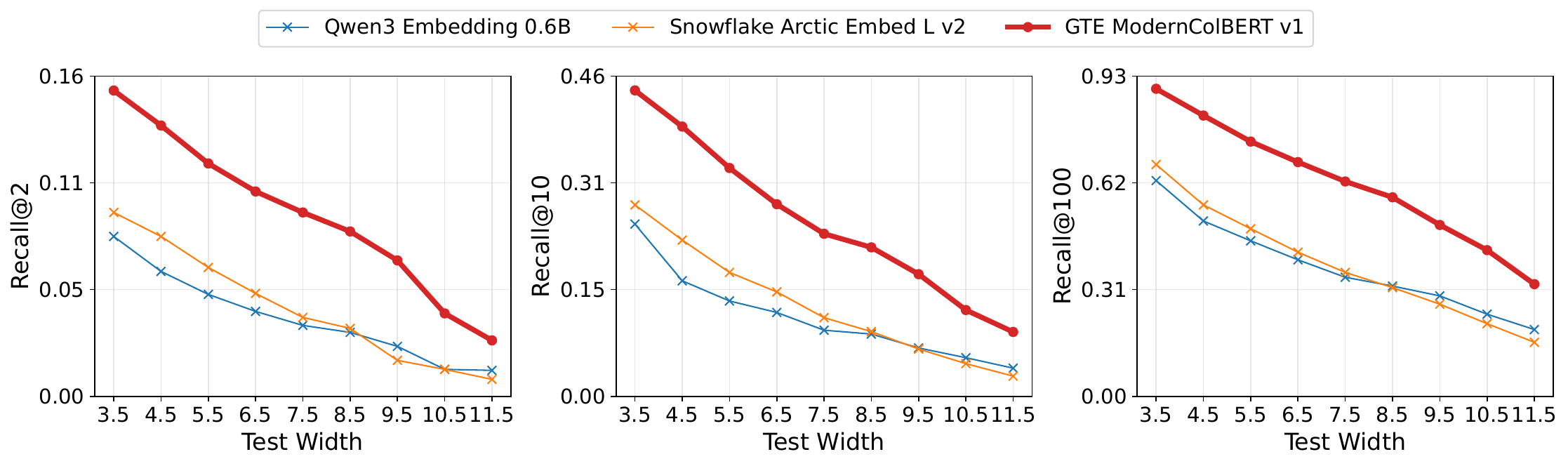}

    \end{subfigure}
    \caption{Retrieval performance across nine test widths at three recall cutoffs, evaluated in the zero-shot setting (top) and after fine-tuning at a training width of $9.5$ (bottom). Multi-vector models consistently outperform single-vector models across all test widths, although the performance of both approaches decreases as the test width increases.}\label{fig:width-sweep}
\end{figure*}

Each run draws from the same family of contrastive loss functions. For a query scored against $D$ documents of which the first $P$ are positive, the loss is

\begin{equation}
  \label{eq:multipos}
  \mathcal{L}
  = \log \sum_{d=1}^{D} e^{s_d/\tau}
    - \frac{1}{P}\sum_{p=1}^{P} \frac{s_p}{\tau},
\end{equation}

where $\tau$ is the temperature and $s_d$ is the similarity score between the query and document $d$. At $P=1$ Equation~\ref{eq:multipos} reduces to softmax cross-entropy against the single positive document, the InfoNCE objective of \citet{ovk18} in the form used for dense retrieval by \citet{koml+20}. At $P=2$ it is the multi-positive extension of \citet{sohn16}, which credits both
positive documents rather than forcing the model to rank one relevant document above another. A detailed exposition on the fine-tuning settings is provided in
Appendix~\ref{sec:app-empirical}.

\subsection{Results}

In this section, we go over the results of our evaluation of the models from Table~\ref{tab:models} on the ANDOR dataset under various settings. We vary zero-shot across various test widths and fine-tuned evaluation across different test and train width to observe the difference in performance between multi-vector and single-vector models. In the main text, we focus our evaluation on recall, and provide additional plots for nDCG and MRR in Appendix~\ref{sec:app-empirical}. In addition, the Appendix includes results for varying training width, as well as results for Jina v4 where the multi and single-vector heads are fine-tuned separately.

\subsubsection{Zero-shot evaluation} We start with pretrained models as a preliminary benchmark. Figure~\ref{fig:width-sweep} considers the gap in recall between GTE ModernColBERT and all the other single-vector models (except Jina). All four single-vector models exhibit substantial gaps with GTE ModernColBERT to varying degrees. Cohere exhibits the strongest performance among the single-vector models at all test widths: GTE ModernColBERT achieves relative gains of about 81\% in Recall@2 and 87\% in Recall@100. All other models are substantially worse, with GTE ModernColBERT achieving consistent multiplicative gains. The margins exhibited by GTE ModernColBERT are noted in Tab~\ref{tab:colbert-relative-gains}. The zero-shot evaluation was also cross checked on known BEIR benchmarks and reproduced the known metrics.
\begin{table}[t]
\centering
\small
\setlength{\tabcolsep}{6pt}
\begin{tabular}{lccc}
\toprule
\multicolumn{4}{c}{\textbf{Zero-shot baselines}} \\
\cmidrule(lr){1-4}
Baseline & Recall@2 & Recall@10 & Recall@100 \\
\midrule
Cohere v4         & 80.6\%   & 93.8\%   & 87.0\%   \\
Qwen3             & 352.2\%  & 346.8\%  & 223.8\%  \\
Snowflake         & 432.3\%  & 396.9\%  & 321.9\%  \\
OAI-3-large       & 1662.1\% & 1204.5\% & 1025.0\% \\
\midrule
Mean gain         & 631.8\%  & 510.5\%  & 414.4\%  \\
\addlinespace
\multicolumn{4}{c}{\textbf{Fine-tuned baselines}} \\
\cmidrule(lr){1-4}
Baseline & Recall@2 & Recall@10 & Recall@100 \\
\midrule
Qwen3      & 98.7\%  & 92.4\% & 55.8\% \\
Snowflake  & 100.1\% & 91.8\% & 59.7\% \\
\midrule
Mean gain  & 99.4\% & 92.1\% & 57.8\% \\
\bottomrule
\end{tabular}
\caption{Relative gains of GTE ModernColBERT over zero-shot and fine-tuned single-vector baselines.}
\label{tab:colbert-relative-gains}

\end{table}
\subsubsection{Fine-tuned evaluation} 
The gap between multi-vector and single-vector models persists despite fine-tuning. At the easiest test width of 3.5, GTE ModernColBERT achieves roughly 89--93\% Recall@100 across training widths. Recall@2 and @10 also show substantial improvements. Despite fine-tuning helping the single-vector models, they do not catch up to GTE ModernColBERT and perform roughly equivalently well at all test widths. We refer the reader to Figure~\ref{fig:width-sweep} for the full range of results, and Tab~\ref{tab:colbert-relative-gains} for relative gains. Averaged over the complete train--test grid and the two single-vector baselines, ColBERT maintains a margin of about 99\% for Recall@2 and 58\% for Recall@100.

\begin{table}[t]
\centering
\small
\setlength{\tabcolsep}{6pt}
\begin{tabular}{lccc}
\toprule
Baseline & Recall@2 & Recall@10 & Recall@100 \\
\midrule
SV fine-tuned     & 104.8\%  & 83.8\%   & 61.5\%  \\
SV zero-shot      & 1613.8\% & 1263.3\% & 767.5\% \\
\bottomrule
\end{tabular}
\caption{Relative gains of jointly fine-tuned Jina Embeddings v4 MV over zero-shot and fine-tuned SV.}
\label{tab:jina-embeddings-relative-gains}
\end{table}

\begin{figure*}[t]
    \centering
    \includegraphics[width=\textwidth]{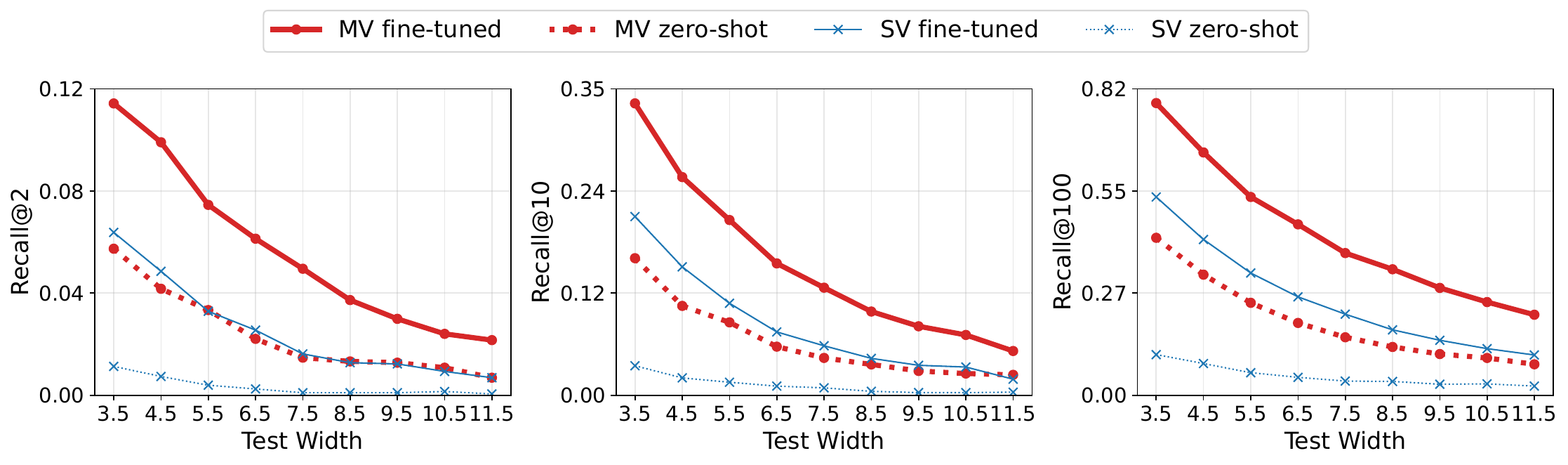}
    \caption{Retrieval performance across nine test widths at three recall cutoffs for Jina Embeddings v4, evaluated after joint fine-tuning at a mean training width of $5.5$.}
    \label{fig:jina-width-sweep}
\end{figure*}
\subsubsection{Simultaneous SV-MV training} The controlled Jina comparison isolates the effect of representation power by updating both the single-vector and multi-vector heads of the same model from the same checkpoint. Figure~\ref{fig:mainfig} and Figure~\ref{fig:jina-width-sweep} show that fine-tuning the single-vector head recovers roughly the ground that zero-shot late interaction already covers, but barely any more.

The fine-tuned single-vector model approaches or modestly exceeds the zero-shot late-interaction baseline rather than clearly surpassing it. The fine-tuned multi-vector model remains substantially ahead throughout, leading the fine-tuned single-vector model at all test widths. Tab~\ref{tab:jina-embeddings-relative-gains} shows relative gains of 105\%, 84\%, and 62\% over the fine-tuned single-vector head at Recall@2, Recall@10, and Recall@100, respectively. While our main text results focus on a training width of 5.5, these trends are replicated across all training widths and additional plots are provided in Appendix~\ref{sec:app-empirical}. This is the strongest evidence we have of a fundamental gap in the representation power of single and multi-vector representations. Among Jina's two output modes, the multi-vector representation also gives stronger zero-shot performance.

\section{Discussion and Limitations}
In this work, we study the retrieval ordering problem and present an explicit family of relevance matrices exhibiting an exponential separation between single-vector and multi-vector representations. Motivated by this result, we introduce ANDOR, a benchmark dataset in the e-commerce domain that captures the same underlying logical structure in a natural semantic retrieval task. We emphasize that ANDOR is not the pattern-matrix instance used in our proof, but rather a semantic instantiation inspired by the theoretical construction. Nevertheless, our experiments show that this semantic benchmark is sufficient to reveal a substantial expressiveness gap between single-vector and multi-vector retrieval models.

Increasing the test width consistently degrades retrieval performance across all methods. One possible explanation is that wider disjunctions require a model to simultaneously represent more alternatives within each category while still enforcing conjunctions across categories. We leave a detailed investigation of this phenomenon to future work.

Finally, we note that each query in ANDOR has only a small number of relevant documents—every test query has exactly two positive documents. Consequently, the construction of \cite{sagks26} guarantees the existence of a $2k+1=5$ dimensional single-vector embedding that exactly realizes the retrieval ordering. However, this is an existential result that assumes access to the complete relevance matrix and does not provide a practical procedure for learning semantic embeddings from queries and documents. In contrast, our experiments operate in the standard retrieval setting, where embeddings must be learned solely from semantic content. Despite the simplicity of the underlying semantic rule and the sparsity of the relevance matrix, existing single-vector models struggle to realize the desired retrieval ordering, whereas multi-vector models are considerably more successful. This suggests that the expressiveness gap identified by our theoretical analysis can also arise in realistic semantic retrieval tasks.

We note down two explicit open problem directions for further investigation.
\begin{asparaitem}
    \item Our proof shows that single vectors fail when we aim to solve the retrieval ordering problem explicitly. In realistic settings, a more relevant question might be to ask a more approximate version of this problem, where Eqn.~\ref{eq:retrievalordering} is satisfied for a $(1-\varepsilon)$ fraction of the entries of each row of the relevance matrix. An interesting future problem would be to check if the exponential gap holds even when our goal is only to solve retrieval order approximately. 
    \item We give an exponential separation between single and multi-vector embeddings when the goal is to rank relevant documents above irrelevant documents. What if the goal is to preserve a specified ranking of documents per query, do the multi-vector embeddings have an exponential advantage in this setting?
\end{asparaitem}

\paragraph{Use of AI Tools} AI tools provided the initial direction and mathematical tools for exploration. The specific constructions and proofs were subsequently entirely done by the authors. Coding agents were further used for coding the dataset and evaluation pipeline. The code was thoroughly reviewed by the authors.

\printbibliography
\appendix
\newpage
\section{ Dataset Construction and Details}\label{sec:app-dataset}

\subsection{Encoding and relevance computation}

As mentioned earlier, each of the 20 categories has 20 attributes, so the set of
values a document can hold in a category is a 20-bit mask. A document is a row of
20 such masks, with zero denoting that the category is absent, and the corpus is
a matrix of shape $50{,}000\times 20$. A query is stored the same way.

Under this encoding the AND-of-OR rule is a single vectorized operation. Let
$Q_c$ denote the query mask and $D_c$ the document mask for category $c$,
\begin{equation}
    \text{rel}(q,d)=1 \iff \forall c\in C_q, \hspace{1cm} Q_c\wedge D_c \ne 0,
\end{equation}
where $\wedge$ is bitwise conjunction. Disjunction within a category is a nonzero
overlap, and conjunction across categories requires that no queried category has
zero overlap. We compute every relevance judgment and collision check with this
expression.

\subsection{Count sampling}

Two kinds of quantity vary in the dataset: how many categories a document or
query holds, and how many attributes each of those categories holds. We control
only two aspects of each: its range and its mean. We therefore draw each count
from the distribution of maximum entropy with the required range and mean. Write
$X$ for the count being sampled, $\{\ell,\dots,h\}$ for the range of integers it
is allowed to take, and $\mu\in(\ell,h)$ for its target mean. This distribution
assigns the probabilities of consecutive values in a constant ratio $r>0$,
\begin{equation}
    \Pr[X=k]\propto r^{k}, \hspace{1cm} k\in\{\ell,\dots,h\},
\end{equation}
so that $r=\Pr[X=k+1]/\Pr[X=k]$ for every $k$ in the range. The mean increases
with $r$, so there is a unique $r$ for which $\mathbb E[X]=\mu$, and we solve for
it numerically.

Whenever $\mu$ is the midpoint of the range, the ratio is $1$ and the
distribution is uniform. The only counts for which it is not are the accept-set
sizes, whose means sit below their midpoints. For test queries, with range
$\{2,\dots,6\}$ and $\mu=3.5$, we get $r=0.77$ and probabilities
$0.31,0.24,0.19,0.14,0.11$ for $2$ through $6$, so a query is nearly three times
as likely to accept two values in a category as six. Most disjunctions are
therefore narrow, which is what keeps the AND-of-OR structure hard.
Table~\ref{tab:app-controls} gives $\ell$, $h$, $\mu$ and $r$ for every count.
 
\begin{table}[h]
\caption{Range and target mean for every sampled count, with the resulting
ratio $r$. The ratio is $1$, and the distribution therefore uniform, wherever
the target mean sits at the midpoint of the range. Accept-set sizes are the sole
exception and are listed at the narrowest width.}
\label{tab:app-controls}
\small
\centering
\begin{tabular}{@{}lccc@{}}
\toprule
Count & Range & $\mu$ & $r$ \\
\midrule
\multicolumn{4}{@{}l}{\emph{Queries}} \\
\quad Categories, test & $\{7,\dots,9\}$ & $8.0$ & $1$ \\
\quad Categories, training & $\{5,\dots,7\}$ & $6.0$ & $1$ \\
\quad Accept-set size, test & $\{2,\dots,6\}$ & $3.5$ & $0.77$ \\
\quad Accept-set size, training & $\{4,\dots,8\}$ & $5.5$ & $0.77$ \\
\addlinespace
\multicolumn{4}{@{}l}{\emph{Documents}} \\
\quad Categories & $\{11,\dots,15\}$ & $13.0$ & $1$ \\
\quad Values per present category & $\{2,\dots,4\}$ & $3.0$ & $1$ \\
\bottomrule
\end{tabular}
\end{table}

\subsection{Query suites and the nested width sweep}
We have nine test suites, each denoted by its query width. The suites share their
queries: a query names the same categories in every suite, and only the number of
values it accepts in each category changes. To keep every test query at exactly
two relevant documents across all nine widths, we widen a query by adding
accepted values and never by replacing them. Concretely, for each query and each
of its categories we fix a random ordering of the twenty values. The narrowest
suite accepts the first few values in that ordering, and each wider suite accepts
one more. The accept sets are therefore nested,
\begin{equation}
    A_{3.5}(c)\subset A_{4.5}(c)\subset\cdots\subset A_{11.5}(c),
\end{equation}
where $A_w(c)$ is the set of values the query accepts in category $c$ at width
$w$.

Nesting lets us state relevance once instead of nine times. A document with an
accepted value at the narrowest width still has one at every wider width, so it
stays relevant throughout; a document with no accepted value even at the widest
width has none at any narrower width, so it stays irrelevant throughout. We use
both directions when building the corpus: each test positive is given a value
from the narrowest accept set, and every other document is checked against the
widest. The same argument fixes hard negatives, since a document built to miss
$k$ categories against the widest accept sets misses those same $k$ at every
width, and we can call it a level-$k$ hard negative without naming a suite.

\subsection{Corpus construction}
\paragraph{Exact category quotas.}
Generating a large pool of candidate documents and keeping only those that
survive collision repair would skew the distribution of the number of categories
per document. We therefore fix that distribution in advance. A quota assigns
exactly $10{,}000$ documents to each category count from $11$ to $15$, is
shuffled once, and then pins the category count of each document as we construct
it. Since the maximum-entropy distribution for this count is uniform, the quota
reproduces it exactly rather than approximating it by sampling.
\paragraph{Test positives.}
For each test query we construct two positive documents. A document satisfies a
query when every queried category holds at least one accepted attribute, so we go
through those categories and place one accepted value in each. We draw that value
from the query's \emph{narrowest} accept set, which by nesting belongs to every
wider accept set as well, so the document stays positive at all nine widths.

We then make the positives look like ordinary documents. We pad each queried
category with further values up to its usual count, and add categories the query
does not name until the document reaches the category count the quota assigned to
it. We also draw the accepted value uniformly rather than always taking the first
value the query lists, since a fixed choice would let a model find positives by
checking a single option instead of evaluating the whole disjunction.
\paragraph{Hard negatives.}
For each test query we manufacture near-misses that satisfy all but $k$ queried
categories, with $k\in\{1,2,3\}$. In the $k$ categories where we want the
document to fail, we give it only values from outside the query's \emph{widest}
accept set, so it fails them at every width. In the remaining categories we give
it an accepted value from the \emph{narrowest} accept set, so it satisfies them at
every width. The document therefore misses the same $k$ categories in all nine
suites, and we can label it level-$k$ once.

A document can also fail a category by not containing that category at all, and
we never do this. Such a negative would carry fewer of the query's categories
than a positive does, so a model could rank documents by counting how many
queried categories appear, and rank correctly without ever looking inside the
disjunction.

Each query with 7 categories receives 24 level-1, 8 level-2 and 4 level-3
documents, and each query with 8 or 9 categories receives 28, 8 and 4. Level 1 is
the hardest case and takes most of the budget, and its share grows with the number
of categories so that each single category is missed by about the same number of
documents.
\paragraph{Distractors.}
We fill the remaining $9{,}220$ positions with random documents, drawn under the
same category-count and value-count controls as positives and hard negatives, so
that a document's role in the construction cannot be read off its shape. Filling
these positions took $10{,}642$ attempts.
\paragraph{Collision repair.}
We check every candidate document against the $1{,}000$ test queries at the
\emph{widest} suite, and retain it only if it satisfies no query other than the
one it was built for. If a document accidentally satisfies another query, we
repair it in place rather than discarding it, since throwing failures away would
bias the corpus toward documents that are easy to construct. We select a
colliding query, and for each of its satisfied categories in random order we
redraw the document's values in that category up to 16 times, preserving the
document's category count and respecting the values the construction requires it
to hold or to avoid. Each redraw can change which \emph{other} queries the
document satisfies, since the values it removes may be the ones that were
matching them, so we score every attempt by the total number of queries the
document still collides with and keep the best one. We then recompute the
collisions and repeat with whatever query remains, for at most 100 rounds. A
document that still collides after that is discarded. Checking at the widest
suite is enough for all nine, because a document that satisfies no query when the
accept sets are largest satisfies none when they are smaller.
\paragraph{Permutation.}
We then randomly permute all $50{,}000$ documents and remap every positive and
hard-negative index, so construction order carries no retrieval signal. Document
identifiers are assigned after the permutation.
\paragraph{Training suites.}
We generate the training queries last, against the frozen corpus, and compute
their relevance from the AND-of-OR rule. We retain all 800 queries, and neither
filter nor resample any of them to hit a target relevance density. The growth in
the number of positives with width in Table~\ref{tab:app-train} is therefore an
effect of widening the accept sets, and not of selection.

\subsection{Retrieval and relevance labels}
For each relevant document, the number of matched values beyond the first
required match is summed across queried categories and normalized by the maximum
count for that query. The grading score is 0.7 times this normalized
extra-overlap score plus 0.3 times a deterministic document-quality score. Up to
the 20 highest-scoring relevant documents receive grade 2, and remaining relevant
documents receive grade 1. Every test query has only two relevant documents, so
both test positives receive grade 2. Dense training queries can have more than 20
positives and therefore contain both grades.
The document-quality score is deterministic rather than random: it is the first 8
bytes of the SHA-256 digest of the string \texttt{seed:doc\_id}, scaled to
$[0,1)$. It is reproducible from the identifier alone, is independent of the
document's attributes, and only breaks ties among documents with equal overlap.

\subsection{Text rendering}
A document is stored as a set of bitmasks corresponding to the category attributes they satisfy. These are not valid inputs to the
models, so we also write each document in a human readable form, as
shown in Figure~\ref{fig:worked-scoring-example}. The text is generated from a
small set of templates and always mentions every attribute the document holds,
without indicating which categories are absent, so nothing about a document's
role can be read off its phrasing or its length. Within a category we join the
accepted attributes with ``or'' and join the categories with ``and'', which is
the AND-of-OR rule read aloud. Queries are addressed to a shopper whose name is
drawn from a fixed pool of first and last names, following the LIMIT
dataset~\cite{wbnl26}, and the same query keeps the same name across all nine
widths.

\subsection{Dataset statistics}

Table~\ref{tab:app-corpus} reports the realized corpus statistics,
Table~\ref{tab:app-hn} the corpus composition, Table~\ref{tab:app-test} the test
suites, and Table~\ref{tab:app-train} the training suites.

\begin{table}[t]
\centering
\small
\begin{tabular}{@{}lr@{}}
\toprule
Statistic & Value \\
\midrule
Documents & 50{,}000 \\
Categories per document & 11--15 \\
Mean categories per document & 13.000 \\
Values per present category & 2--4 \\
Mean values per present category & 2.978 \\
Total values per document & 24--55 \\
Mean total values per document & 38.710 \\
\bottomrule
\end{tabular}
\caption{Realized corpus statistics. The category-count distribution is exactly
balanced by construction, with 10{,}000 documents at each count from 11 to 15.}
\label{tab:app-corpus}

\end{table}

\begin{table}[t]
\centering
\small
\begin{tabular}{@{}lr@{}}
\toprule
Construction role & Documents \\
\midrule
Designated test positives & 2{,}000 \\
Injected hard negatives & 38{,}780 \\
\quad level 1 (one category missed) & 26{,}780 \\
\quad level 2 (two categories missed) & 8{,}000 \\
\quad level 3 (three categories missed) & 4{,}000 \\
Random distractors & 9{,}220 \\
\midrule
Total & 50{,}000 \\
\bottomrule
\end{tabular}
\caption{Corpus composition by construction role. A level-$k$ hard negative
violates exactly $k$ queried categories at every width. Roles describe
construction only; serialized records carry no role field, since relevance is
query-dependent.}\label{tab:app-hn}

\end{table}

\begin{table}[t]
\centering
\small
\begin{tabular}{@{}lrrrrr@{}}
\toprule
Suite & \makecell{Accepted values\\per category} & Realized mean & \makecell{Total accepted\\per query} & Mean total & Positives \\
\midrule
3.5  & 2--6   & 3.506  & 16--41  & 28.166 & 2 \\
4.5  & 3--7   & 4.506  & 23--50  & 36.199 & 2 \\
5.5  & 4--8   & 5.506  & 30--59  & 44.232 & 2 \\
6.5  & 5--9   & 6.506  & 37--68  & 52.265 & 2 \\
7.5  & 6--10  & 7.506  & 44--77  & 60.298 & 2 \\
8.5  & 7--11  & 8.506  & 51--86  & 68.331 & 2 \\
9.5  & 8--12  & 9.506  & 58--95  & 76.364 & 2 \\
10.5 & 9--13  & 10.506 & 65--104 & 84.397 & 2 \\
11.5 & 10--14 & 11.506 & 72--113 & 92.430 & 2 \\
\bottomrule
\end{tabular}
\caption{Test suites. All nine share one corpus, the same 1{,}000 query
skeletons, and the same category counts: 305 queries with 7 categories, 357 with
8, and 338 with 9, for a mean of 8.033. Every test query has exactly two relevant
documents at every width, so difficulty is not confounded by a change in
relevance cardinality.}
\label{tab:app-test}

\end{table}

\begin{table}[t]
\centering
\small
\begin{tabular}{@{}lrrrrrr@{}}
\toprule
Suite & \makecell{Accepted values\\per category} & Realized mean & \makecell{Total accepted\\per query} & Mean total & Positives & Mean positives \\
\midrule
5.5 & 4--8  & 5.491 & 20--49 & 32.945 & 11--1{,}114  & 208.104     \\
6.5 & 5--9  & 6.491 & 25--56 & 38.945 & 36--1{,}571  & 378.955     \\
7.5 & 6--10 & 7.491 & 30--63 & 44.945 & 95--2{,}028  & 605.574     \\
8.5 & 7--11 & 8.491 & 35--70 & 50.945 & 180--2{,}584 & 880.398     \\
9.5 & 8--12 & 9.491 & 40--77 & 56.945 & 294--3{,}076 & 1{,}185.989 \\
\bottomrule
\end{tabular}
\caption{Training suites used for fine-tuning. All share the same 800 query
skeletons, with exactly 267 queries at 5 categories, 266 at 6, and 267 at 7.
Relevance is computed exactly against the frozen corpus and no skeleton is
filtered by density, so the growth in positives is the effect of widening the
accept sets.}\label{tab:app-train}

\end{table}

\subsection{Category and attribute vocabulary}

Table~\ref{tab:app-vocab} lists the fixed vocabulary of 20 categories with 20
attribute values each. For readability, the table renders the released-file
identifiers in title case with underscores replaced by spaces.
\begin{table}[p]
\rotatebox{90}{%
\begin{minipage}{0.9\textheight}
\centering
\scriptsize
\begin{tabular}{lp{20cm}}
\toprule
Category & Attribute values \\
\midrule
Brand & Aria Wear, Loom Lab, Nova Stitch, Urban Thread, Casa Cloth, Vivid Line, Metro Loft, Cotton Cove, Highland Mode, Echo Apparel, Amber Lane, Silver Loom, Blue Ridge, Prism Attire, Cedar Fashion, Moss and Mill, Velvet Arc, Northbay Style, Riverloom, Solara Studio \\
Color & Black, White, Navy, Charcoal, Cream, Beige, Brown, Red, Burgundy, Pink, Lavender, Purple, Blue, Teal, Green, Olive, Yellow, Orange, Silver, Gold \\
\makecell[l]{Garment\\ Type} & T Shirt, Shirt, Blouse, Sweater, Cardigan, Hoodie, Jacket, Coat, Dress, Skirt, Trousers, Jeans, Shorts, Blazer, Vest, Jumpsuit, Kurta, Tunic, Leggings, Scarf \\
Material & Cotton, Linen, Wool, Silk, Denim, Polyester, Nylon, Rayon, Viscose, Leather, Suede, Fleece, Cashmere, Bamboo, Modal, Spandex, Canvas, Corduroy, Satin, Chiffon \\
Texture & Smooth, Ribbed, Waffle, Brushed, Quilted, Pleated, Sheer, Structured, Crinkled, Soft, Rugged, Lightweight, Heavyweight, Stretch, Breathable, Insulated, Matte, Glossy, Embroidered, Beaded \\
Fit & Slim, Regular, Relaxed, Oversized, Tailored, Cropped, Longline, High Waist, Low Rise, Straight, Tapered, Bootcut, Flared, Boxy, Athletic, Maternity, Petite, Plus, Loose, Compression \\
Size Range & XS, Small, Medium, Large, XL, XXL, Petite Size, Tall Size, Plus Size, One Size, Kids, Toddler, Junior, Maternity Size, Unisex Size, EU 36, EU 38, EU 40, EU 42, EU 44 \\
\makecell[l]{Sleeve\\ Length} & Sleeveless, Short Sleeve, Cap Sleeve, Elbow Sleeve, Three Quarter Sleeve, Long Sleeve, Raglan Sleeve, Balloon Sleeve, Puff Sleeve, Bishop Sleeve, Dolman Sleeve, Cuffed Sleeve, Rolled Sleeve, Adjustable Sleeve, Detachable Sleeve, Strapless, Spaghetti Strap, Halter, One Shoulder, Kimono Sleeve \\
\makecell[l]{Neckline\\ Collar} & Crew Neck, V Neck, Scoop Neck, Boat Neck, Turtleneck, Mock Neck, Collared, Mandarin Collar, Peter Pan Collar, Shawl Collar, Notch Lapel, Hooded, Square Neck, Sweetheart Neck, Halter Neck, Wrap Neck, Off Shoulder, Henley Neck, Zip Neck, Polo Collar \\
Closure & Pullover, Button Front, Zip Front, Snap Button, Hook Eye, Tie Front, Wrap Tie, Drawstring, Elastic Waist, Buckle, Velcro, Toggle, Magnetic Closure, Hidden Zip, Side Zip, Back Zip, Lace Up, Clasp, No Closure, Double Breasted \\
Pattern & Solid, Striped, Checked, Plaid, Floral, Geometric, Polka Dot, Animal Print, Paisley, Abstract, Colorblock, Tie Dye, Ombre, Herringbone, Houndstooth, Camo, Graphic, Embroidered Pattern, Jacquard, Lace Pattern \\
Occasion & Casual, Office, Party, Wedding, Travel, Workout, Lounge, Beach, Festival, Formal, Semi Formal, Outdoor, School, Date Night, Interview, Ceremony, Hiking, Commute, Dinner, Resort \\
Season & Spring, Summer, Autumn, Winter, Monsoon, Year Round, Holiday, Festive, Transitional, Cold Weather, Hot Weather, Humid Weather, Dry Weather, Windy Weather, Evening, Daytime, Vacation, Back to School, Festive Winter, Rainy Day \\
Weather & Breathable Hot, Insulated Cold, Water Resistant, Waterproof, Windproof, Quick Dry, Moisture Wicking, Thermal, UV Protective, Lightweight Layer, Heavy Layer, Rain Ready, Snow Ready, Humid Friendly, Wrinkle Resistant, Stain Resistant, Odor Resistant, Packable, Sun Safe, Cozy \\
Price Tier & Budget, Value, Midrange, Premium, Luxury, Outlet, Sale, Under 25, Under 50, Under 100, Investment, Affordable Luxury, Designer, Student Friendly, Workwear Basic, Capsule Wardrobe, Limited Edition, Handcrafted, Imported, Local Maker \\
Care & Machine Wash, Hand Wash, Dry Clean, Tumble Dry Low, Line Dry, Wrinkle Free, Iron Low, No Iron, Delicate Cycle, Colorfast, Pre Shrunk, Wash Cold, Spot Clean, Bleach Free, Easy Care, Quick Clean, Steam Only, Wash Inside Out, Dry Flat, Low Maintenance \\
Sustainability & Organic, Recycled, Fair Trade, Vegan, Cruelty Free, Low Water, Natural Dye, Biodegradable, Upcycled, Carbon Neutral, Locally Made, Repairable, Durable, Plastic Free, Responsibly Sourced, Slow Fashion, Ethically Made, Renewable Fiber, Secondhand, Zero Waste \\
Style & Minimalist, Classic, Bohemian, Streetwear, Sporty, Preppy, Vintage, Modern, Elegant, Edgy, Romantic, Professional, Chic, Rustic, Glam, Utilitarian, Artsy, Casual Style, Formal Style, Resort Style \\
Target User & Women, Men, Unisex, Teens, Kids, Toddlers, Maternity User, Petite User, Tall User, Plus User, Athletes, Professionals, Travelers, Students, Seniors, Dancers, Hikers, Office Workers, Party Goers, Minimalists \\
\makecell[l]{Accessory \\Compatibility }& Belt Friendly, Scarf Pairing, Blazer Layering, Sneaker Match, Boot Match, Heel Match, Handbag Match, Cap Match, Jewelry Pairing, Watch Pairing, Backpack Friendly, Travel Bag Friendly, Gym Bag Friendly, Raincoat Layer, Cardigan Layer, Coat Layer, Leggings Pairing, Denim Pairing, Formal Shoe Match, Sandal Match \\
\bottomrule
\end{tabular}
\caption{The complete ANDOR vocabulary: 20 categories, 20 attribute values each.}
\label{tab:app-vocab}
\end{minipage}
}
\end{table}

\section{Additional Experimental Results}\label{sec:app-empirical}

This appendix reports the variation of train and test width across the models, together with
nDCG and MRR views of the same runs. Every fine-tuned result uses the final
checkpoint at step $6250$ and is evaluated on the same nine test suites, whose
mean widths range from $3.5$ to $11.5$. Five training widths from $5.5$ to $9.5$ runs are done, making a grid over all $45$ training-test combinations for each model. 

\subsection{Model Settings}

Figure~\ref{fig:app-finetuned-widths-recall} complements
Figure~\ref{fig:width-sweep} with the four Recall training widths not shown
there. The nDCG and MRR figures report all five training widths.

\begin{figure*}[p]
    \centering
    \includegraphics[width=\textwidth,height=0.82\textheight,keepaspectratio]{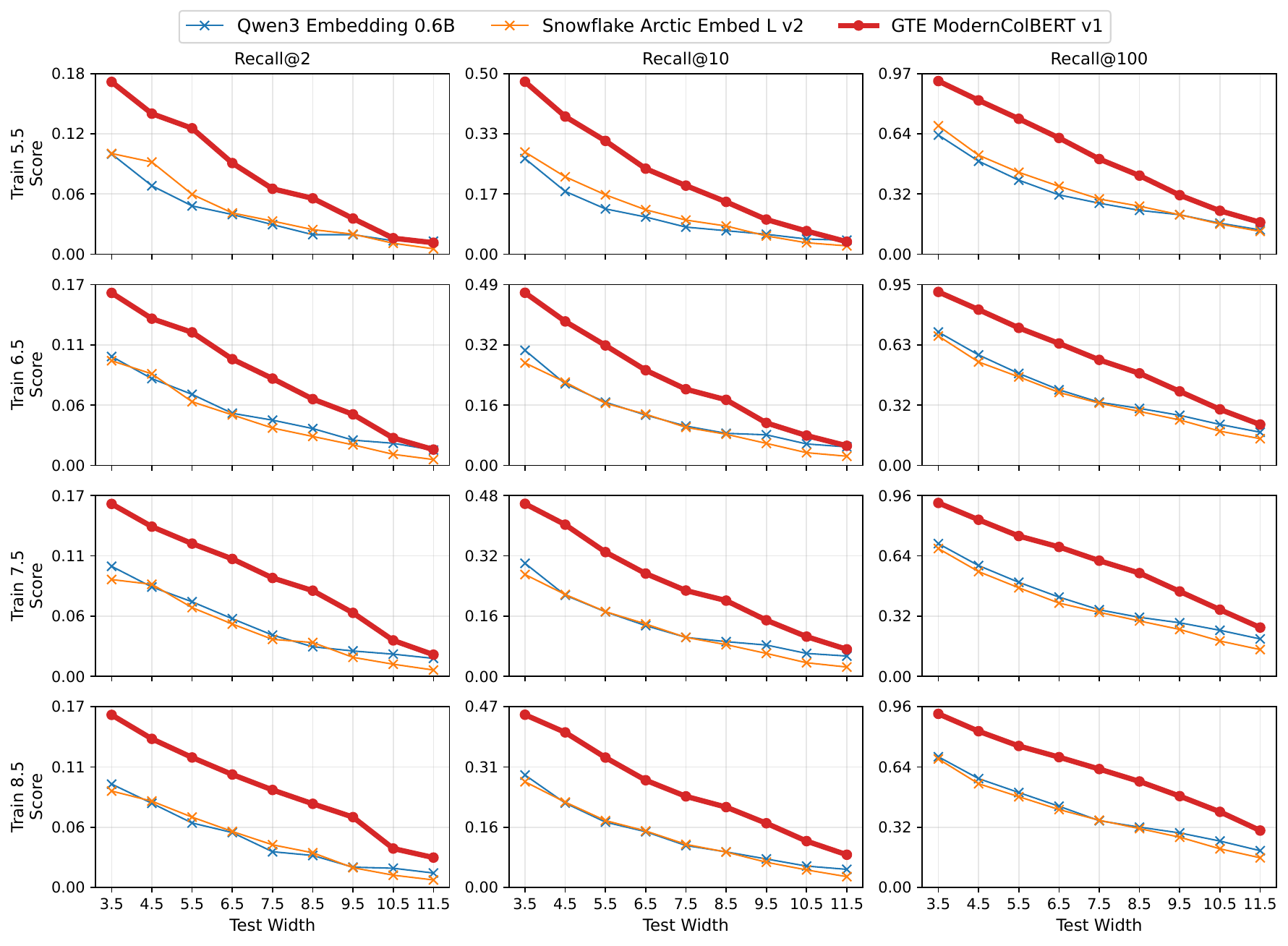}
    \caption{Recall across test mean widths for fine-tuned GTE ModernColBERT v1, Qwen3 Embedding 0.6B, and Snowflake Arctic Embed L v2 at training mean widths $5.5$, $6.5$, $7.5$, and $8.5$. The corresponding results at training width $9.5$ appear in Figure~\ref{fig:width-sweep}.}
    \label{fig:app-finetuned-widths-recall}
\end{figure*}

\subsection{Training paradigms}\label{sec:app-paradigms}

Section~\ref{sec:experiments} states the hyper-parameters shared by every fine-tuning
run: the $800$ training queries, the $6{,}250$ optimizer steps, the effective
batch of $32$ queries, the $32$ candidates scored per query, and the composition
of the hard-negative pool. This section records what distinguishes the runs from
one another. We use five paradigms, listed in Table~\ref{tab:app-paradigms},
whose objectives and optimization settings appear in
Table~\ref{tab:app-objective} and whose systems settings appear in
Table~\ref{tab:app-systems}. Every run uses AdamW with a linear schedule and
$313$ warmup steps, query and document truncation lengths of $320$ and $512$
tokens, an effective batch of $32$ queries. All
reported numbers come from the final checkpoint at step $6{,}250$.

\begin{table}[t]
\centering
\small
\setlength{\tabcolsep}{4pt}
\begin{tabular}{@{}llc@{}}
\toprule
Paradigm & Model & Scoring \\
\midrule
P1 Single-vector & Qwen3 Embedding 0.6B  & Cosine \\
P2 Single-vector & Snowflake Arctic L v2 & Cosine \\
P3 Late interaction & GTE ModernColBERT v1 & MaxSim \\
P4 Joint SV--MV  & Jina Embeddings v4    & Both \\
P5 Separate SV, MV & Jina Embeddings v4  & Both \\
\bottomrule
\end{tabular}
\caption{The five fine-tuning paradigms. Each is run once per training width,
giving five runs per paradigm. P5 trains two independent models per width.}
\label{tab:app-paradigms}
\end{table}

\begin{table}[t]
\centering
\small
\setlength{\tabcolsep}{4pt}
\begin{tabular}{@{}lccccc@{}}
\toprule
Run & $P$ & $\tau$ & LR & WD & Clip \\
\midrule
P1 Qwen3      & $2$ & $0.05$  & $3$  & $0.01$ & $64$ \\
P2 Snowflake  & $2$ & $0.05$  & $3$  & $0.01$ & $64$ \\
P3 ColBERT    & $1$ & $1/320$ & $10$ & $0$    & $1$  \\
P4 Jina joint & $2$ & $0.05$  & $3$  & $0.01$ & --- \\
P5 Jina SV    & $2$ & $0.05$  & $3$  & $0.01$ & --- \\
P5 Jina MV    & $1$ & $1/320$ & $10$ & $0.01$ & --- \\
\bottomrule
\end{tabular}
\caption{Objective and optimization settings. $P$ is the number of supervised
positives in Equation~\ref{eq:multipos}; all runs score $D=32$ candidates.
Learning rates are given in units of $10^{-7}$. The temperature $1/320$ is the
reciprocal of the query-token budget and reproduces unnormalized summed MaxSim
logits.}
\label{tab:app-objective}
\end{table}

\begin{table}[t]
\centering
\small
\setlength{\tabcolsep}{4pt}
\begin{tabular}{@{}lcccc@{}}
\toprule
Run & Precision & Batch & Accum. & Ckpt.\ grad. \\
\midrule
P1 Qwen3      & bf16 & $8$ & $4$    & Yes \\
P2 Snowflake  & bf16 & $8$ & $4$    & Yes \\
P3 ColBERT    & fp16 & $1$ & $32$   & No  \\
P4 Jina joint & bf16 & $1$ & $32/G$ & No  \\
P5 Jina SV    & bf16 & $1$ & $32/G$ & No  \\
P5 Jina MV    & bf16 & $1$ & $32/G$ & No  \\
\bottomrule
\end{tabular}
\caption{Systems settings. The global effective batch size, computed as the
number of devices times the per-device batch times the accumulation count, is
$32$ in every run; $G$ is the number of visible devices,
which is one, two, or four for the Jina runs.}
\label{tab:app-systems}
\end{table}

\paragraph{P1 and P2: single-vector baselines.}
Qwen3 Embedding 0.6B and Snowflake Arctic Embed L v2 are fine-tuned with the
multi-positive objective of Equation~\ref{eq:multipos} at $P=2$, so both
designated positive documents of a training query contribute to the numerator. Queries
and documents are embedded independently and scored by cosine similarity. Both models expect an asymmetric encoding convention, so queries carry the
instruction prefix each model was pretrained with, Qwen3 receiving a task
instruction and Snowflake the string \texttt{query:}, while documents are
encoded unprefixed. 

\paragraph{P3: late interaction.}
GTE ModernColBERT v1 retains one vector per token and is scored with the MaxSim
sum of Equation~\ref{eq:Chamfer}. Because that sum grows with query length, we
divide it by the number of query tokens before the softmax, which makes the
logit scale independent of query length, and we set the temperature to the
reciprocal of the $320$-token query budget so that the resulting logits match
the unnormalized summed scores the released recipe assumes. This run uses the
one-positive objective, $P=1$. 

\paragraph{P4: joint single- and multi-vector training.}
Jina Embeddings v4 exposes a single-vector head and a multi-vector head over one
shared backbone, so both can be updated from the same checkpoint, on the same
batches, under one optimizer. The objective is the unweighted sum of three terms
computed on the same $32$ candidates: a Matryoshka single-vector loss that
averages Equation~\ref{eq:multipos} over the truncation dimensions $128$, $256$,
$512$, $1024$, and $2048$; the same loss applied to length-normalized MaxSim
scores from the multi-vector head; and a Kullback--Leibler term distilling the
full $2048$-dimensional single-vector score distribution into the
late-interaction distribution. Both distributions are computed at the same
temperature and neither side is detached, so this term couples the two heads
symmetrically rather than performing a one-directional teacher--student update.
The heads are then evaluated separately from the shared checkpoint. This is the
tightest control we have, because the representations differ only in how the
backbone output is consumed.

\paragraph{P5: separately trained single- and multi-vector models.}
 P5 trains the single-vector and multi-vector representations as independent runs, each with the settings appropriate to its scoring function, as recorded in
Table~\ref{tab:app-objective}.

\subsection{Qwen, Snowflake and GTE-Modern Colbert Results}

\paragraph{Recall.}
Recall decreases monotonically with test width throughout the complete five-width grid. Averaged over training widths, GTE ModernColBERT falls from
$91.4\%$ to $25.4\%$ Recall@100 between test widths $3.5$ and $11.5$, a loss of
$66.0$ pp and a factor of $3.6$. The decline is steeper at shallow cutoffs:
Recall@10 and Recall@2 decrease by factors of $6.8$ and $7.9$, respectively.
Qwen3 falls from $67.2\%$ to $17.8\%$ Recall@100, while Snowflake falls from
$68.0\%$ to $14.4\%$. The two single-vector models are close on the narrow and
middle suites, but Qwen3 retains a $3.4$ pp advantage over Snowflake at test
width $11.5$.

GTE ModernColBERT leads the stronger single-vector model at every test width.
The size of that lead depends on how it is measured. In absolute terms the gap
is widest on the easier widths, reaching $27.3$ pp of Recall@100 at test width
$4.5$, and it narrows to $7.6$ pp at width $11.5$. Measured as a ratio, the lead instead grows
steadily as the task hardens, from $1.34\times$ at width $3.5$ to a maximum of
$1.75\times$ at width $7.5$, because the single-vector models degrade faster
than late interaction; it falls back to $1.43\times$ only at the widest suite.

\begin{figure*}[p]
    \centering
    \includegraphics[width=\textwidth,height=0.82\textheight,keepaspectratio]{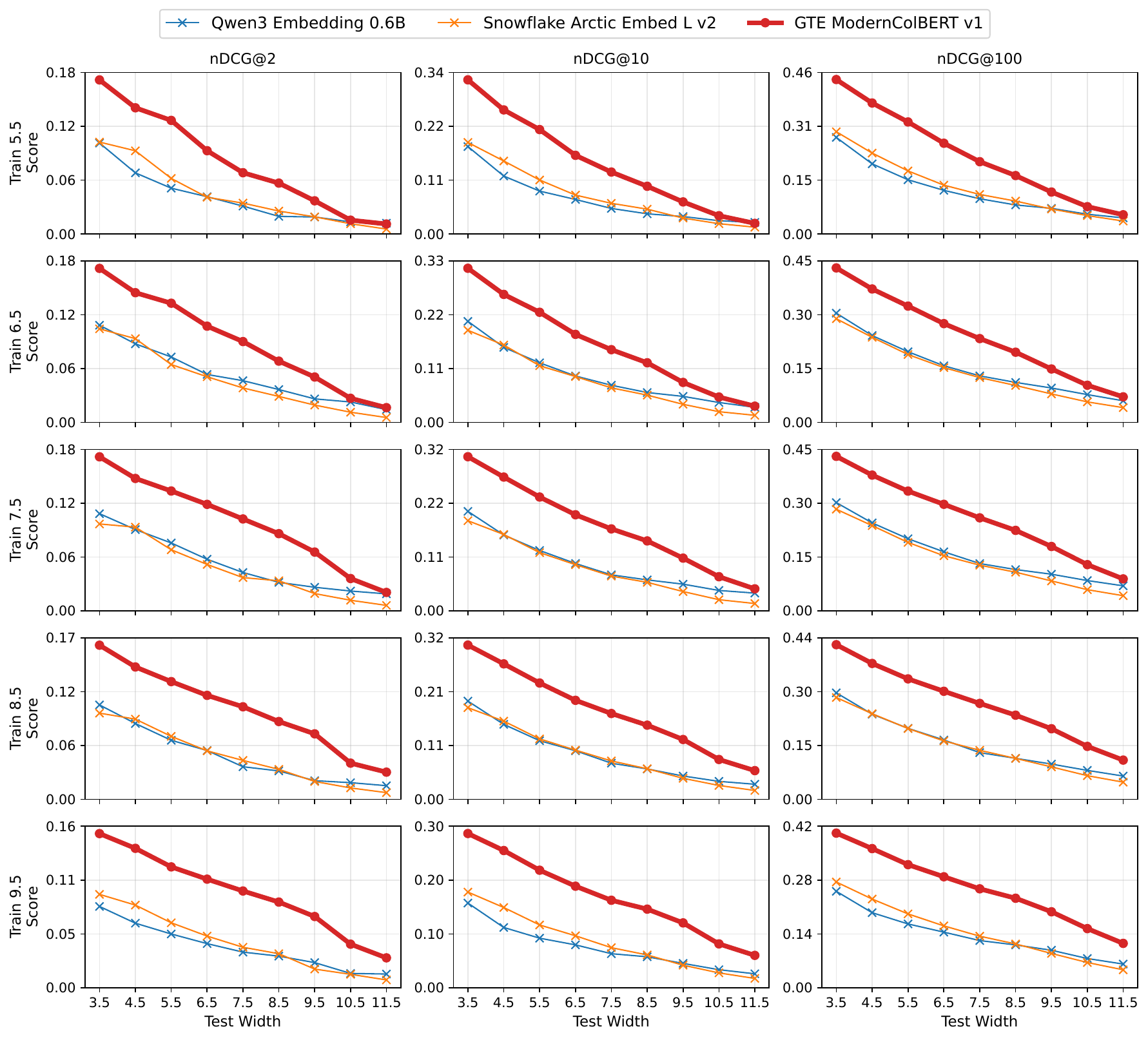}
    \caption{nDCG@2, nDCG@10, and nDCG@100 across test mean widths for fine-tuned GTE ModernColBERT v1, Qwen3 Embedding 0.6B, and Snowflake Arctic Embed L v2 at all five training mean widths from $5.5$ to $9.5$.}
    \label{fig:app-finetuned-widths-ndcg}
\end{figure*}

\paragraph{nDCG.}
Figure~\ref{fig:app-finetuned-widths-ndcg} shows the same degradation after
discounting relevant documents by rank. Since each test query has exactly two
relevant documents, nDCG@2 is numerically close to Recall@2 and stays within
$8\%$ of it across the grid. The deeper nDCG@100 view is more informative. From
test width $3.5$ to $11.5$, it decreases by factors of $4.9$, $4.7$, and $6.6$
for ColBERT, Qwen3, and Snowflake, respectively. Thus, relevant documents found
at deep cutoffs on wide suites increasingly occur at poor ranks.  

\begin{figure*}[p]
    \centering
    \includegraphics[width=\textwidth,height=0.82\textheight,keepaspectratio]{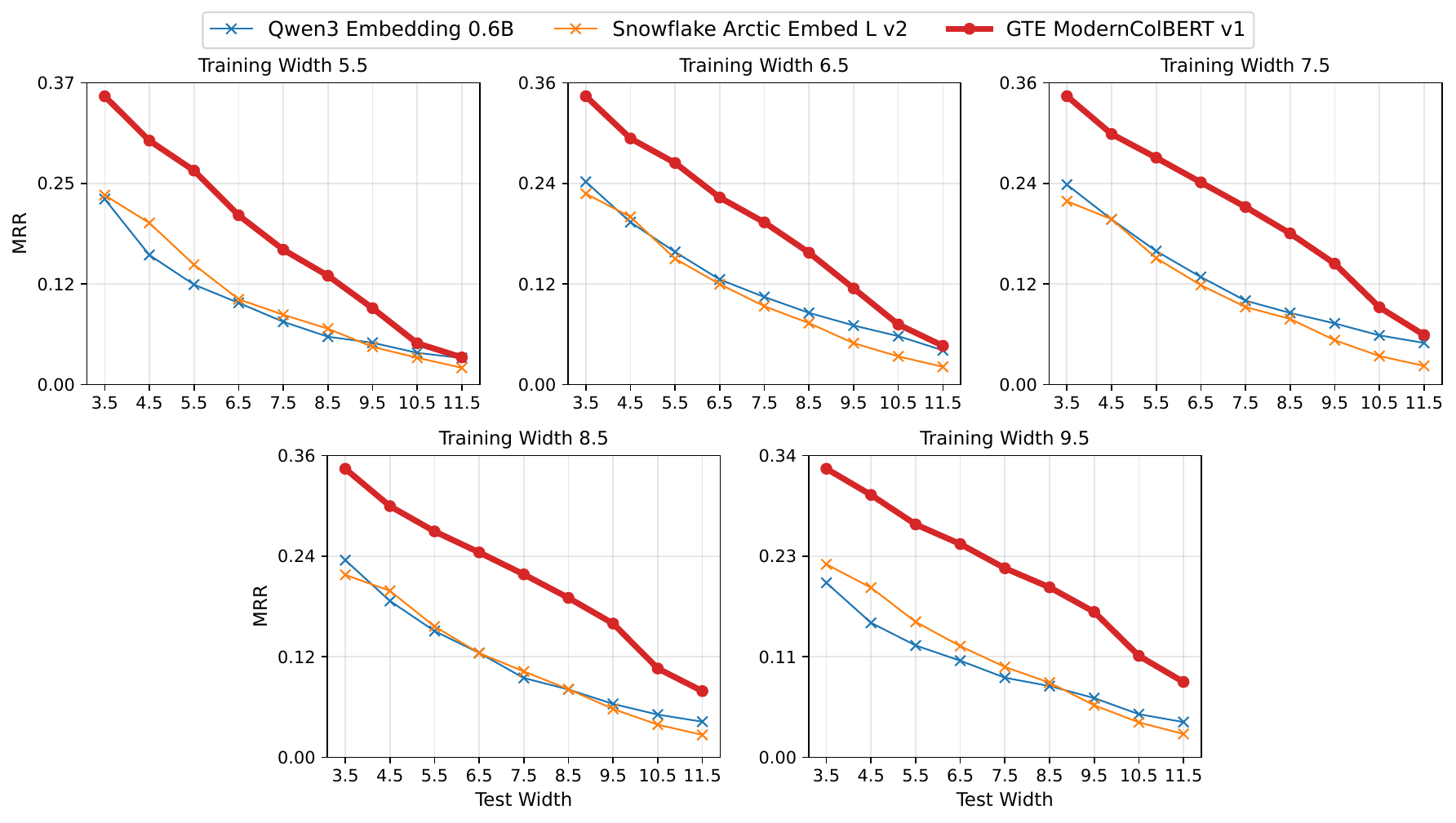}
    \caption{MRR across test mean widths for fine-tuned GTE ModernColBERT v1, Qwen3 Embedding 0.6B, and Snowflake Arctic Embed L v2 at all five training mean widths from $5.5$ to $9.5$.}
    \label{fig:app-finetuned-widths-mrr}
\end{figure*}

\paragraph{MRR.}
MRR is the strictest of the three views. It is also the view in which the two families are
separated most cleanly: pooled over the train--test grid, late interaction leads
the single-vector models by roughly $80\%$ relatively. All three models decay
monotonically as the suites widen, and Snowflake decays fastest, losing an order
of magnitude between the narrowest and widest suites while the other two lose about a factor of six.

\subsection{Variation across training widths}

\begin{figure*}[p]
    \centering
    \includegraphics[width=\textwidth,height=0.82\textheight,keepaspectratio]{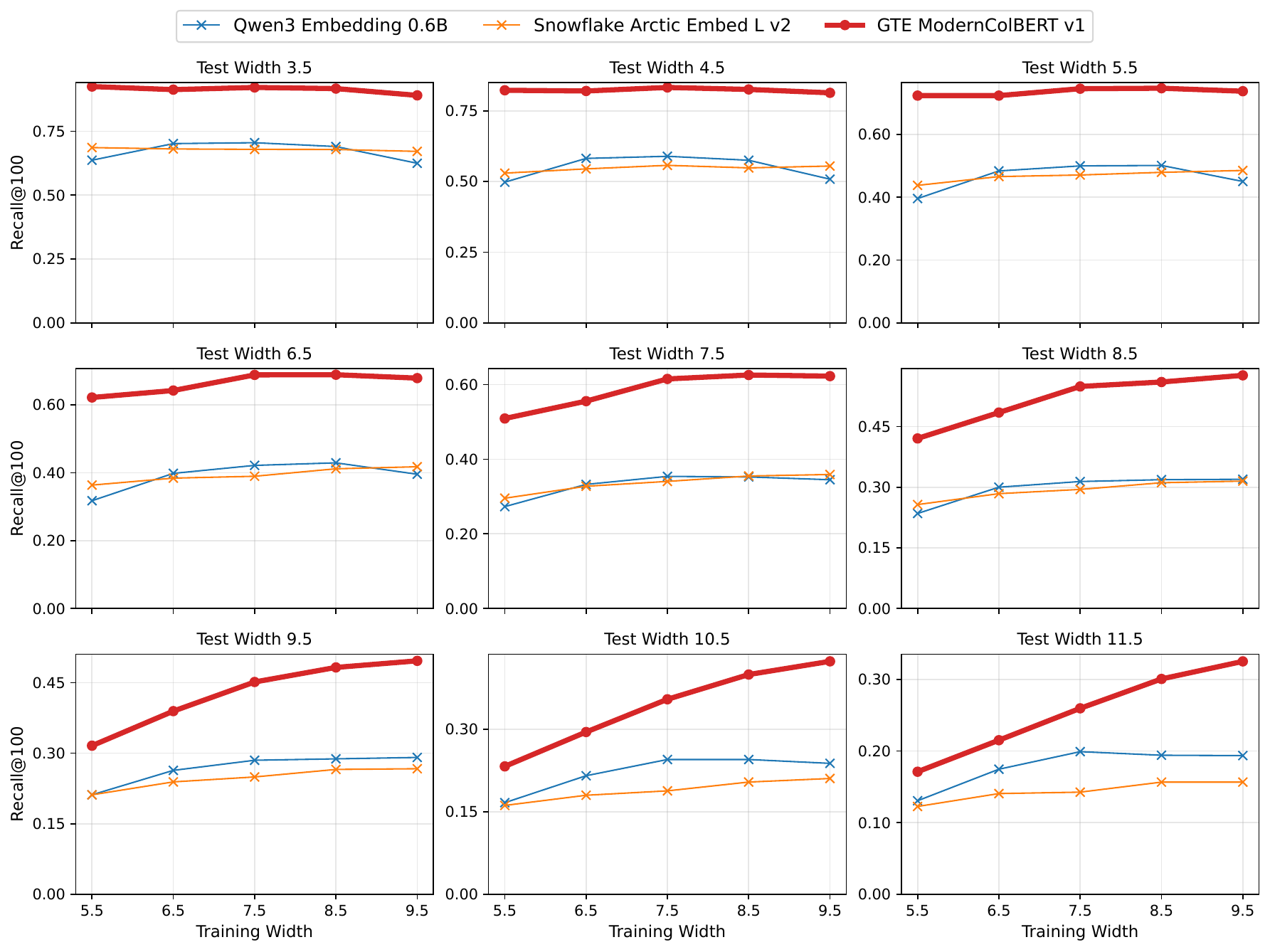}
    \caption{Recall@100 as a function of training mean width for fine-tuned GTE ModernColBERT v1, Qwen3 Embedding 0.6B, and Snowflake Arctic Embed L v2. The nine panels cover every test mean width from $3.5$ to $11.5$.}
    \label{fig:app-training-width-sweep}
\end{figure*}

Figure~\ref{fig:app-training-width-sweep} shows the  sweep of training widths holding the test width fixed and varying the training width. The effect is small on smaller test widths but substantial on wide ones. For ColBERT, the five
training widths span only $3.45$ pp at test width $3.5$, from $92.55\%$ when
trained at width $5.5$ to $89.10\%$ when trained at $9.5$. At test width
$11.5$, however, Recall@100 rises monotonically from $17.10\%$ to $32.50\%$, a
$15.40$ pp or $90\%$ relative improvement. Training on wider queries therefore
substantially improves robustness to wide test suites at little cost on narrow
ones.

ColBERT and Snowflake improve monotonically with training width, peaking at
$9.5$, whereas Qwen3 has an optimum at $7.5$ and is the most sensitive
of the three, so matching training width to test width is not generally optimal.
The ordering is nevertheless unchanged in every panel: ColBERT stays above both
single-vector models at every training and test width.

\subsection{Jointly and Separately trained Jina representations}
Figures~\ref{fig:app-jina-joint-widths-recall}--\ref{fig:app-jina-joint-widths-mrr}
report Jina Embeddings v4 with both representations trained jointly, and
Figures~\ref{fig:app-jina-separate-widths-recall}--\ref{fig:app-jina-separate-widths-mrr}
repeat the comparison with the two trained independently. Recall for the joint
setting at training width $5.5$ appears in Figure~\ref{fig:jina-width-sweep}.

\begin{figure*}[p]
    \centering
    \includegraphics[width=\textwidth,height=0.82\textheight,keepaspectratio]{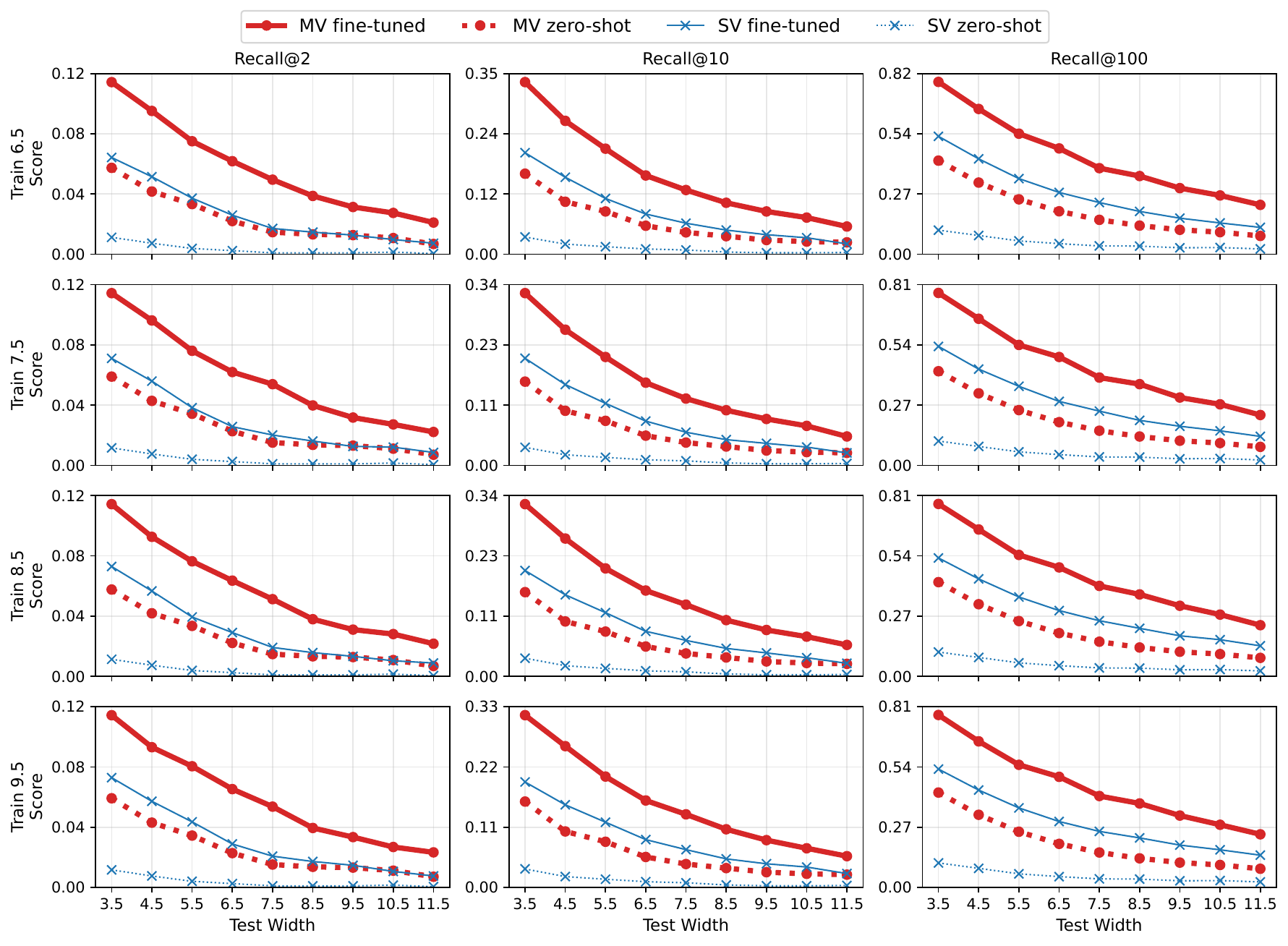}
    \caption{Recall across test mean widths for the jointly trained Jina Embeddings v4 single-vector and multi-vector heads at training mean widths $6.5$, $7.5$, $8.5$, and $9.5$. Solid lines are fine-tuned results and dotted lines are the corresponding zero-shot baselines. Results at training width $5.5$ appear in Figure~\ref{fig:jina-width-sweep}.}
    \label{fig:app-jina-joint-widths-recall}
\end{figure*}

\begin{figure*}[p]
    \centering
    \includegraphics[width=\textwidth,height=0.82\textheight,keepaspectratio]{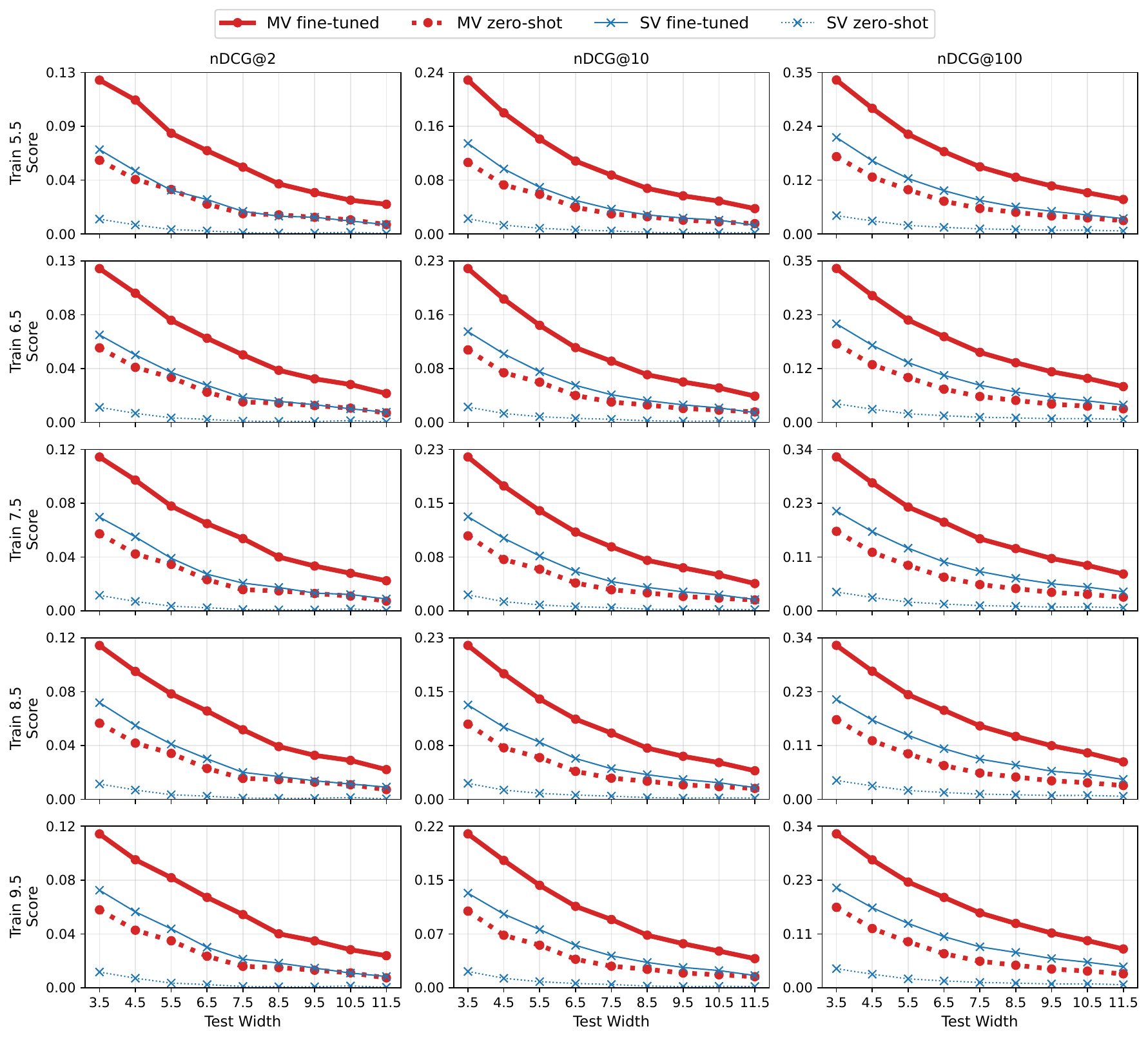}
    \caption{nDCG@2, nDCG@10, and nDCG@100 across test mean widths for the jointly trained Jina Embeddings v4 single-vector and multi-vector heads at all five training mean widths from $5.5$ to $9.5$. Solid lines are fine-tuned results and dotted lines are the corresponding zero-shot baselines.}
    \label{fig:app-jina-joint-widths-ndcg}
\end{figure*}

\begin{figure*}[p]
    \centering
    \includegraphics[width=\textwidth,height=0.82\textheight,keepaspectratio]{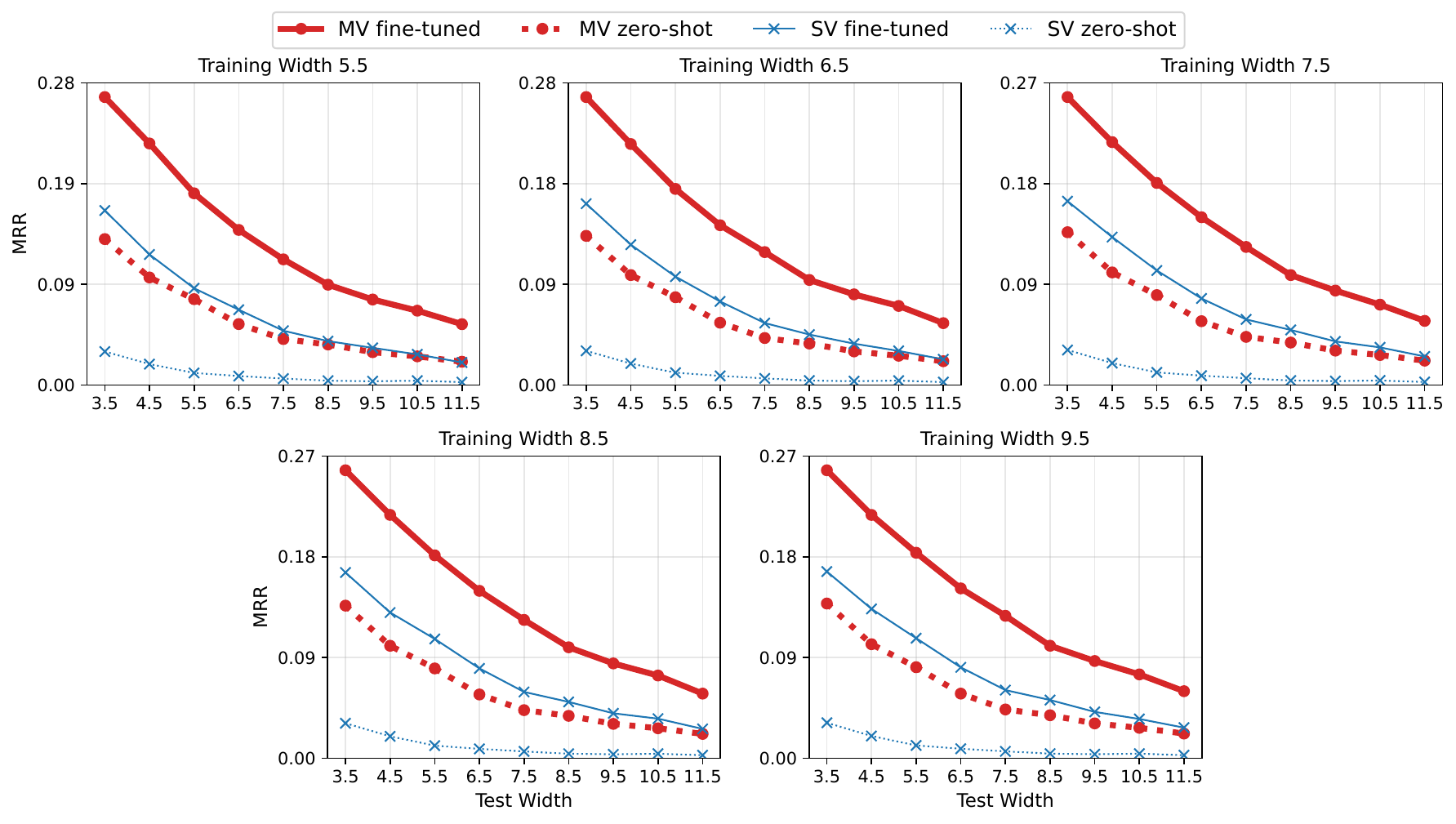}
    \caption{MRR across test mean widths for the jointly trained Jina Embeddings v4 single-vector and multi-vector heads at all five training mean widths from $5.5$ to $9.5$. Solid lines are fine-tuned results and dotted lines are the corresponding zero-shot baselines.}
    \label{fig:app-jina-joint-widths-mrr}
\end{figure*}

\begin{figure*}[p]
    \centering
    \includegraphics[width=\textwidth,height=0.82\textheight,keepaspectratio]{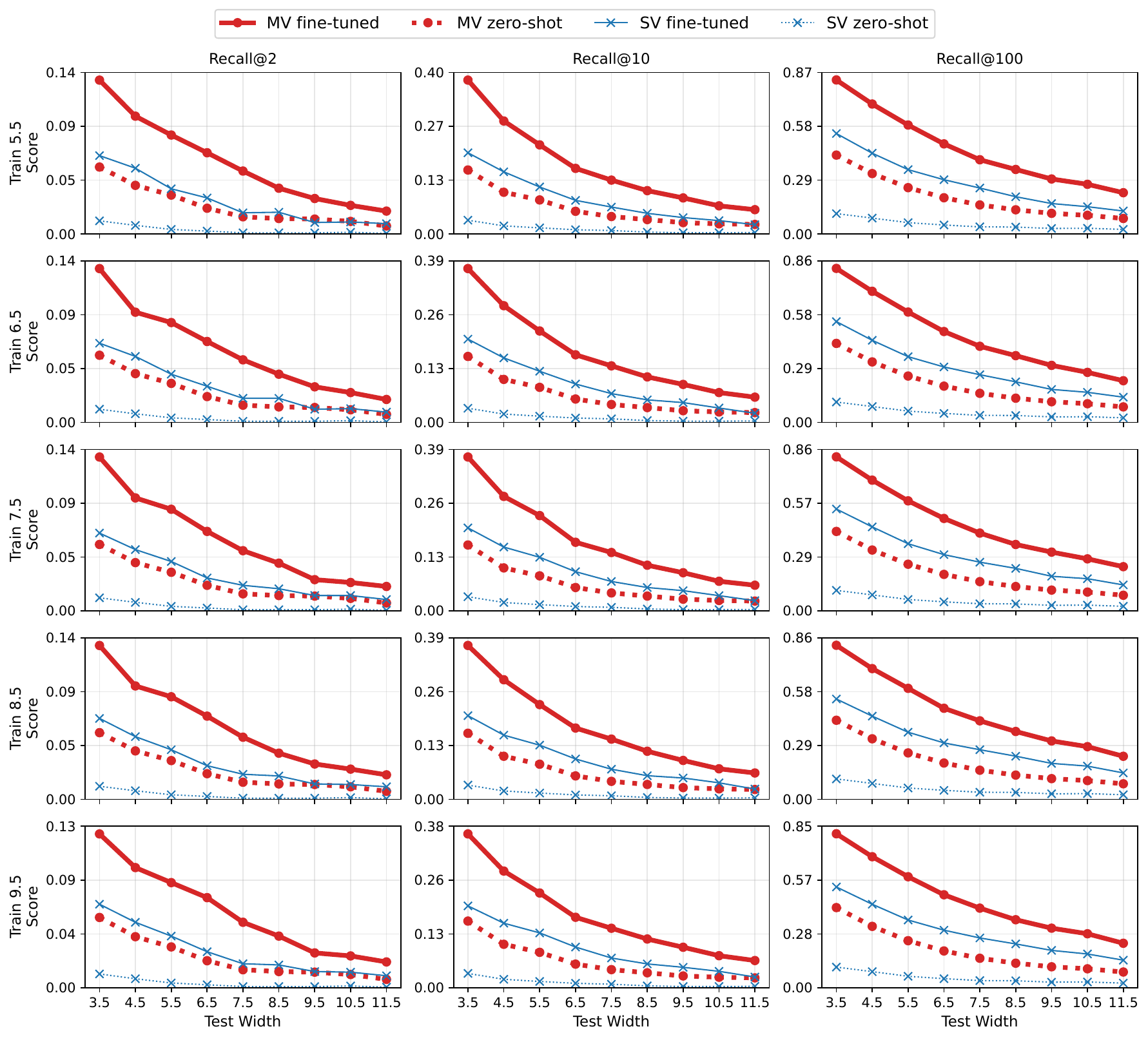}
    \caption{Recall across test mean widths for Jina Embeddings v4 single-vector and multi-vector models trained separately at all five training mean widths. Solid lines are fine-tuned results and dotted lines are the corresponding zero-shot baselines.}
    \label{fig:app-jina-separate-widths-recall}
\end{figure*}

\begin{figure*}[p]
    \centering
    \includegraphics[width=\textwidth,height=0.82\textheight,keepaspectratio]{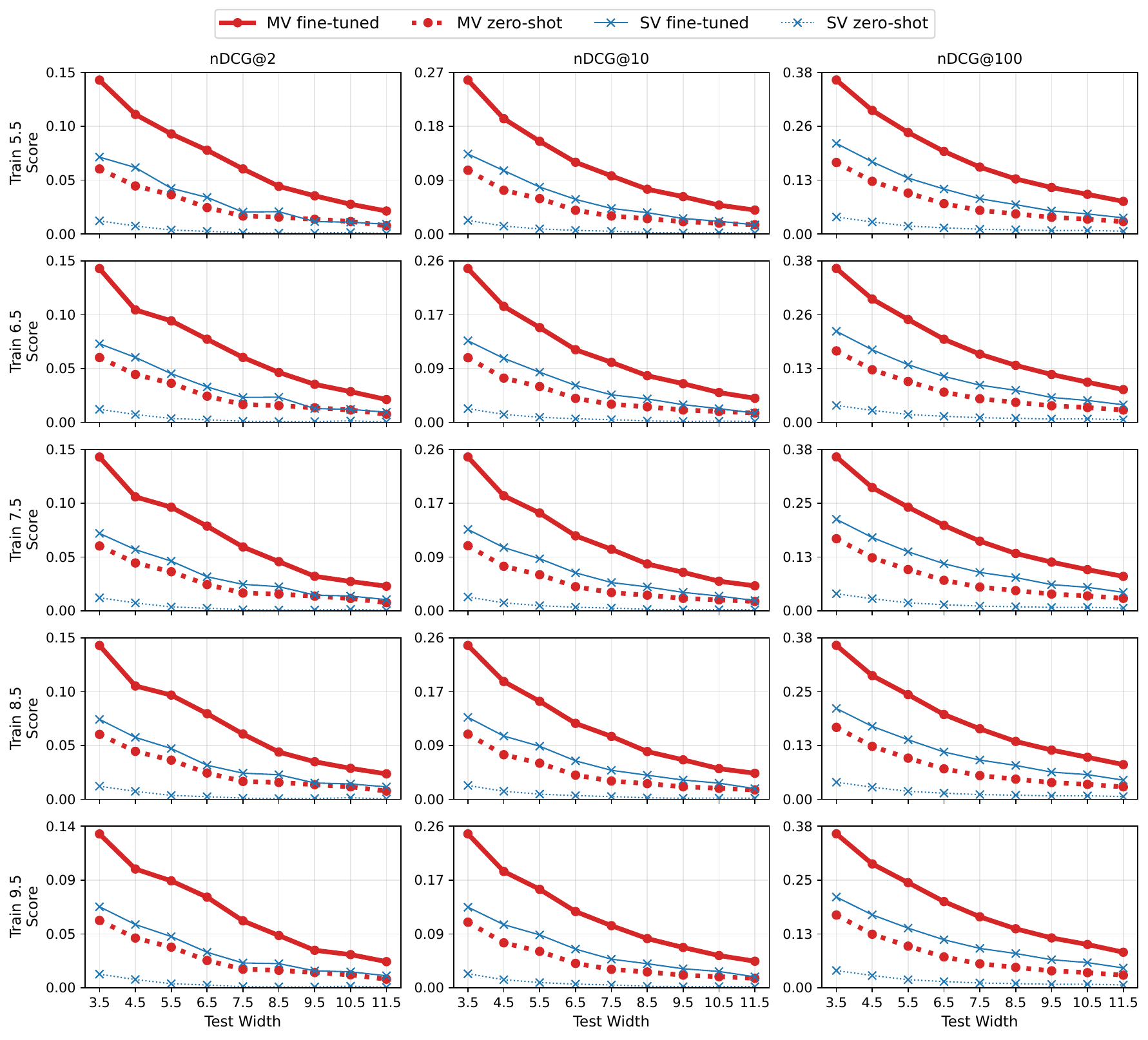}
    \caption{nDCG@2, nDCG@10, and nDCG@100 across test mean widths for Jina Embeddings v4 single-vector and multi-vector models trained separately at all five training mean widths. Solid lines are fine-tuned results and dotted lines are the corresponding zero-shot baselines.}
    \label{fig:app-jina-separate-widths-ndcg}
\end{figure*}

\begin{figure*}[p]
    \centering
    \includegraphics[width=\textwidth,height=0.82\textheight,keepaspectratio]{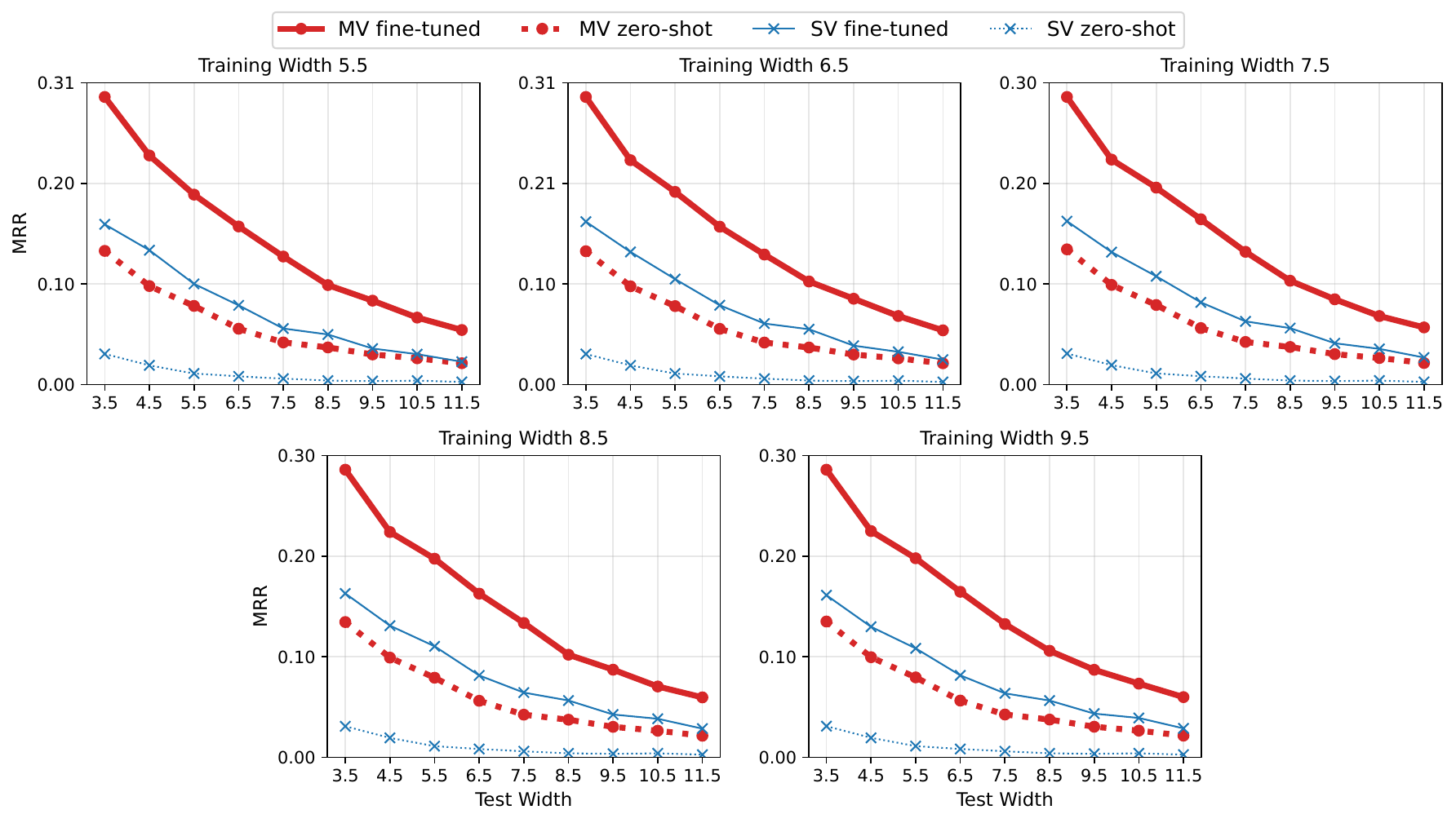}
    \caption{MRR across test mean widths for Jina Embeddings v4 single-vector and multi-vector models trained separately at all five training mean widths. Solid lines are fine-tuned results and dotted lines are the corresponding zero-shot baselines.}
    \label{fig:app-jina-separate-widths-mrr}
\end{figure*}

\paragraph{Consistency with the cross-model comparison.}
Both Jina settings reproduce, on a single backbone, every trend established for
the three separately pretrained models. Recall, nDCG and MRR fall monotonically
as the test suites widen; the multi-vector representation leads the
single-vector one at all $45$ train--test combinations and at every cutoff; the
lead is widest at shallow cutoffs, roughly doubling Recall@2 while adding about
$60\%$ at Recall@100, with nDCG and MRR in between; the relative lead grows as
the suites harden; and training width has only a marginal effect, smaller for
the multi-vector representation than for the single-vector one. What the shared
backbone adds is a controlled reading of what supervision buys, since the dotted
zero-shot curves and the solid fine-tuned curves belong to the same model.
Fine-tuning lifts the single-vector representation far more in relative terms,
because it starts from a much weaker baseline, but the absolute gain shrinks
steadily as the suites widen. The decisive comparison is between the fine-tuned
single-vector representation and the \emph{zero-shot} multi-vector one: the two
are within roughly a percentage point of each other at Recall@2 across the
entire sweep, and the fine-tuned single-vector representation pulls clearly
ahead only at Recall@100. Task-specific supervision therefore buys the
single-vector representation approximately what untrained late interaction
already delivers at the ranks that matter most, while the fine-tuned
multi-vector representation stays well ahead of both. 
\paragraph{Joint versus independent training.}
Training the two representations jointly or independently changes their absolute
quality slightly and their separation almost not at all. Independent training is
modestly better for both, raising grid-mean Recall@100 by about $4\%$ relative
for each, with somewhat larger gains at Recall@10. The ratio between the
representations, however, is essentially unmoved: pooled over the grid, the
multi-vector-to-single-vector ratios are $2.07$, $1.89$ and $1.61$ at Recall@2,
Recall@10 and Recall@100 under independent training, against $2.05$, $1.84$ and
$1.61$ under joint training, with nDCG@100 and MRR likewise differing by only a
couple of points of relative margin. This matters because the two settings fail
in opposite directions: joint training equalizes data order, optimizer state and
parameter count but denies each representation its own recipe, and it even
distills the single-vector ranking into the late-interaction head, whereas
independent training gives each its preferred recipe but reintroduces those
nuisance factors. The gap survives both, so it is not an artifact of sharing a
training run or of a particular choice of objective. We nonetheless read the two
as independent measurements rather than a clean ablation, since the independent
multi-vector runs are the corrected reruns and therefore differ from the joint
runs in temperature and learning rate as well as in training strategy.

\end{document}